\documentclass[twoside,11pt]{article}
\usepackage[utf8]{inputenc}%(only for the pdftex engine)
\usepackage[hyper]{jmlr}
\usepackage{tikz-cd} 
\jmlrheading{2019}{Ashish Dandekar, Debabrota Basu and St\'ephane Bressan}
\ShortHeadings{Differential Privacy at Risk}{Dandekar, Basu and Bressan}
\firstpageno{1}

\begin{document}
 
\title{Differential Privacy at Risk: Bridging Randomness and Privacy Budget}

\author{\name Ashish Dandekar \email adandekar@ens.fr\\
	\addr D{\'e}partement d’informatique,\\
	{\'E}cole Normale Sup{\'e}rieure\\
	Paris, France
	\AND
	\name Debabrota Basu \email basud@chalmers.se\\
	\addr Data Science and AI Division\\
	 Department of Computer Science and Engineering\\
	 Chalmers University of Technology\\
	 G{\"o}teborg, Sweden
	 \AND
	 \name St\'ephane Bressan \email  steph@nus.edu.sg\\
	 \addr School of Computing\\
	 National University of Singapore\\
	 Singapore, Singapore}

\maketitle

\begin{abstract}%
The calibration of noise for a privacy-preserving mechanism depends on the sensitivity of the query and the prescribed privacy level. A data steward must make the non-trivial choice of a privacy level that balances the requirements of users and the monetary constraints of the business entity.
\\
We analyse roles of the sources of randomness, namely the explicit randomness induced by the noise distribution and the implicit randomness induced by the data-generation distribution, that are involved in the design of a privacy-preserving mechanism. The finer analysis enables us to provide stronger privacy guarantees with quantifiable risks. Thus, we propose \textit{privacy at risk} that is a probabilistic calibration of privacy-preserving mechanisms. We provide a composition theorem that leverages privacy at risk. We instantiate the probabilistic calibration for the Laplace mechanism by providing analytical results.
\\
We also propose a cost model that bridges the gap between the privacy level and the compensation budget estimated by a GDPR compliant business entity. The convexity of the proposed cost model leads to a unique fine-tuning of privacy level that minimises the compensation budget. We show its effectiveness by illustrating a realistic scenario that avoids overestimation of the compensation budget by using privacy at risk for the Laplace mechanism. We quantitatively show that composition using the cost optimal privacy at risk provides stronger privacy guarantee than the classical advanced composition.
\end{abstract}
 
\section{Introduction}
\label{sec:intro}
\citeauthor{dwork2014algorithmic} quantify the privacy level $\epsilon$ in \emph{$\epsilon$-differential privacy} as an upper bound on the worst-case privacy loss incurred by a privacy-preserving mechanism. Generally, a privacy-preserving mechanism perturbs the results by adding the calibrated amount of random noise to them. The calibration of noise depends on the sensitivity of the query and the specified privacy level. In a real-world setting, a data steward must specify a privacy level that balances the requirements of the users and monetary constraints of the business entity. \citeauthor{garfinkel2018issues} report the issues in deploying differential privacy as the privacy definition by the US census bureau. They highlight the lack of analytical methods to choose the privacy level.  They also report empirical studies that show the loss in utility due to the application of privacy-preserving mechanisms.

We address the dilemma of a data steward in two ways. Firstly, we propose a probabilistic quantification of privacy levels. Probabilistic quantification of privacy levels provides a data steward a way to take quantified risks under the desired utility of the data. We refer to the probabilistic quantification as \emph{privacy at risk}. We also derive a composition theorem that leverages privacy at risk. Secondly, we propose a cost model that links the privacy level to a monetary budget. This cost model helps the data steward to choose the privacy level constrained on the estimated budget and vice versa. Convexity of the proposed cost model ensures the existence of a unique privacy at risk that would minimise the budget. 
We show that the composition with an optimal privacy at risk provides stronger privacy guarantees than the traditional advanced composition~\citep{dwork2014algorithmic}.
In the end, we illustrate a realistic scenario that exemplifies how the data steward can avoid overestimation of the budget by using the proposed cost model by using privacy at risk.  

%{\color{red} This part is a bit hand wavy without using risk terminology.}
The probabilistic quantification of privacy levels depends on two sources of randomness: the \emph{explicit randomness} induced by the noise distribution and the \emph{implicit randomness} induced by the data-generation distribution. Often, these two sources are coupled with each other. We require analytical forms of both sources of randomness as well as an analytical representation of the query to derive a privacy guarantee. Computing the probabilistic quantification is generally a challenging task. Although we find multiple probabilistic privacy definitions in the literature~\citep{machanavajjhala2008privacy,hall2012random}, we are missing analytical quantification bridging the randomness and privacy level of a privacy-preserving mechanism. To the best of our knowledge, we are the first to analytically derive such a probabilistic quantification, namely privacy at risk, for the widely used Laplace mechanism~\citep{dwork2006calibrating}. We also derive a composition theorem with privacy at risk. It is a special case of the advanced composition theorem~\citep{dwork2014algorithmic} that deals with a sequential and adaptive use of privacy-preserving mechanisms. We work on a simpler model independent evaluations used in the basic composition theorem~\citep{dwork2014algorithmic}.

The privacy level proposed by the differential privacy framework is too abstract a quantity to be integrated in a business setting. We propose a cost model that maps the privacy level to a monetary budget. The corresponding cost model for the probabilistic quantification of privacy levels is a convex function of the privacy level. Hence, it leads to a unique probabilistic privacy level that minimises the cost. We illustrate a realistic scenario in a GDPR compliant business entity that needs an estimation of the compensation budget that it needs to pay to stakeholders in the unfortunate event of a personal data breach. The illustration shows that the use of probabilistic privacy levels avoids overestimation of the compensation budget without sacrificing utility.

In this work,  we comparatively evaluate the privacy guarantees using privacy at risk of using Laplace mechanism. 
We quantitatively compare the composition  under the optimal privacy at risk, which is estimated using the cost model, with traditional composition mechanisms - the basic composition and advanced mechanism~\citep{dwork2014algorithmic}. We observe that it gives stronger privacy guarantees than the ones by the advanced composition without sacrificing on the utility of the mechanism.

In conclusion, benefits of the probabilistic quantification i.e. the privacy at risk are twofold. It not only quantifies the privacy level for a given privacy-preserving mechanism but also facilitates decision-making in problems that focus on the privacy-utility trade-off and the compensation budget minimisation.

\section{Background}
\label{sec:background}
We consider a universe of datasets $\mathcal{D}$. We explicitly mention when we consider that the datasets are sampled from a data-generation distribution $\mathcal{G}$ with support $\mathcal{D}$. Two datasets of equal cardinality $x$ and $y$ are said to be \textit{neighbouring datasets} if they differ in one data point. A pair of neighbouring datasets is denoted by $x \sim y$. In this work, we focus on a specific class of queries called \textit{numeric queries}. A numeric query $f$ is a function that maps a dataset into a real-valued vector, i.e.  $f: \mathcal{D} \rightarrow \mathbb{R}^k$. For instance, a sum query returns the sum of the values in a dataset.

In order to achieve a privacy guarantee, a \emph{privacy-preserving mechanism}, or \textit{mechanism} in short, is a randomised algorithm, that adds noise to the query from a given family of distributions. Thus, a privacy-preserving mechanism of a given family, $\mathcal{M}(f, \Theta)$, for the query $f$ and the set of parameters $\Theta$ of the given noise distribution, is a function that maps a dataset into a real vector, i.e.  $\mathcal{M}(f, \Theta): \mathcal{D} \rightarrow \mathbb{R}^k$. We denote a privacy-preserving mechanism as $\mathcal{M}$, when the query and the parameters are clear from the context.

% \subsection{Differential privacy}

\begin{definition}{\textbf{[Differential Privacy~\citep{dwork2014algorithmic}.]}}
A privacy-preserving mechanism $\mathcal{M}$, equipped with a query $f$ and with parameters $\Theta$, is $(\epsilon, \delta)$-differentially private if for all $Z \subseteq \textrm{Range}(\mathcal{M})$ and $x, y \in \mathcal{D}$ such that $x \sim y$:
% \[
% \left | \log{\frac{\mathbb{P}(\mathcal{M}(f, \Theta)(x) \in Z)}{\mathbb{P}(\mathcal{M}(f, \Theta)(y) \in Z)}} \right| \leq \epsilon .
% \]
\[
\mathbb{P}(\mathcal{M}(f, \Theta)(x) \in Z) \leq \text{e}^\epsilon \mathbb{P}(\mathcal{M}(f, \Theta)(y) \in Z) + \delta 
\]
\label{def:DP}
$(\epsilon, 0)$-differentially private mechanism is ubiquitously called as $\epsilon$-differentially private.
\end{definition}

A privacy-preserving mechanism provides perfect privacy if it yields indistinguishable outputs for all neighbouring input datasets. The privacy level $\epsilon$ quantifies the privacy guarantee provided by $\epsilon$-differential privacy.  For a given query, the smaller the value of the $\epsilon$, the qualitatively higher is the privacy.  A randomised algorithm that is $\epsilon$-differentially private is also $\epsilon'$-differential private for any $\epsilon' > \epsilon$.

In order to satisfy $\epsilon$-differential privacy, the parameters of a privacy-preserving mechanism requires a calculated calibration. The amount of noise required to achieve a specified privacy level depends on the query. If the output of the query does not change drastically for two neighbouring datasets, then small amount of noise is required to achieve a given privacy level. The measure of such fluctuations is called the \emph{sensitivity} of the query. The parameters of a privacy-preserving mechanism are calibrated using the sensitivity of the query that quantifies the smoothness of a numeric query.

\begin{definition}{\textbf{[Sensitivity.]}}\label{def:sensitivity}
The sensitivity of a query $f: \mathcal{D} \rightarrow \mathbb{R}^k$ is defined as
\[
\Delta_f \triangleq \max_{\substack{x,y \in \mathcal{D} \\  x \sim y}} ~\lVert f(x) - f(y) \rVert_1.
\]
\end{definition}
% \subsection{Laplace mechanism}

The Laplace mechanism is a privacy-preserving mechanism that adds scaled noise sampled from a calibrated Laplace distribution to the numeric query.
\begin{definition}\textit{[\citep{papoulis2002probability}]}
The Laplace distribution with mean zero and scale $b > 0$ is a probability distribution with probability density function
\[
\mathrm{Lap}(b) \triangleq \frac{1}{2b}\exp{(-\frac{|x|}{b})},
\]
where $x \in \mathbb{R}$. We write $\mathrm{Lap}(b)$ to denote a random variable $X \sim \mathrm{Lap}(b)$
\end{definition}

\begin{definition}{\textbf{[Laplace Mechanism~\citep{dwork2006calibrating}.]}}
Given any function $f:\mathcal{D} \rightarrow \mathbb{R}^k$ and any $x \in \mathcal{D}$, the Laplace Mechanism is defined as
\[
\mathcal{L}^{\Delta_f}_{\epsilon}(x) \triangleq \mathcal{M}\left(f, ~\frac{\Delta_f}{\epsilon}\right)(x) =  f(x) + (L_1, ..., L_k),
\]
where $L_i$ is drawn from $\mathrm{Lap}\left(\frac{\Delta_f}{\epsilon}\right)$ and added to the $i^{\mathrm{th}}$ component of $f(x)$.
\label{def:lap_mech}
\end{definition}

\begin{theorem}~\citep{dwork2006calibrating}
The Laplace mechanism, $\mathcal{L}^{\Delta_f}_{\epsilon_0}$, is $\epsilon_0$-differentially private.
\label{thm:dp}
\end{theorem}
% In order to satisfy $\epsilon$-differential privacy, the calibration of the Laplacian noise depends on the sensitivity of the query and the desired privacy level $\epsilon$. Note that for the value of sensitivity $\Delta$ smaller than $\Delta_f$, $\mathcal{L}_{\epsilon_0}^\Delta$ is not $\epsilon_0$-differentially private.

\section{Privacy at Risk: A Probabilistic Quantification of Randomness}
\label{sec:privacy_at_risk}

The parameters of a privacy-preserving mechanism are calibrated using the privacy level and the sensitivity of the query. A data steward needs to choose appropriate privacy level for practical implementation. \citeauthor{lee2011much} show that the choice of an actual privacy level by a data steward in regard to her business requirements is a non-trivial task. Recall that the privacy level in the definition of differential privacy corresponds to the worst case privacy loss. Business users are however used to taking and managing risks, if the risks can be quantified. For instance, \citeauthor{var_book} defines \textit{Value at Risk} that is used by risk analysts to quantify the loss in investments for a given portfolio and an acceptable confidence bound. Motivated by the formulation of \textit{Value at Risk}, we propose to use the use of probabilistic privacy level. It provides us a finer tuning of an $\epsilon_0$-differentially private privacy-preserving mechanism for a specified risk $\gamma$.
% we define \textit{privacy at risk} as a privacy definition. For a given privacy-preserving mechanism, privacy at risk defines the  privacy level $\epsilon$ with a confidence level $\gamma$. For the sake of clarity, we refer to this privacy level $\epsilon$ as the privacy at risk level.

% \begin{defi}[\textbf{Privacy at risk}]
% For a given data generating distribution $\mathcal{G}$, a privacy-preserving mechanism $\mathcal{M}$, equipped with a query $f$ and with parameters $\Theta$, satisfies $(\epsilon, \gamma)$-privacy at risk, if for all $Z \subseteq \textrm{Range}(\mathcal{M})$ and $x, y$ sampled from $\mathcal{G}$ such that $x \sim y$:
% \begin{equation}\label{eqn:par}
%     \mathbb{P}\left[\log{\left |\frac{\mathbb{P}(\mathcal{M}(f, \Theta)(x) \in Z)}{\mathbb{P}(\mathcal{M}(f, \Theta)(y) \in Z)}\right|} > \epsilon \right] \leq \gamma,
% \end{equation}
% where the outer probability is calculated with respect to the probability space $\mathrm{Range}(\mathcal{M} \circ \mathcal{G})$ obtained by applying the privacy-preserving mechanism $\mathcal{M}$ on the data-generation distribution $\mathcal{G}$.
% \label{def:par}
% \end{defi}

\begin{definition}{\textbf{[Privacy at Risk.]}}
For a given data generating distribution $\mathcal{G}$, a privacy-preserving mechanism $\mathcal{M}$, equipped with a query $f$ and with parameters $\Theta$, satisfies $\epsilon$-differential privacy with a \emph{privacy at risk} $0 \leq \gamma \leq 1$, if for all $Z \subseteq \textrm{Range}(\mathcal{M})$ and $x, y$ sampled from $\mathcal{G}$ such that $x \sim y$:
\begin{equation}\label{eqn:par}
    \mathbb{P}\left[\left | \ln{\frac{\mathbb{P}(\mathcal{M}(f, \Theta)(x) \in Z)}{\mathbb{P}(\mathcal{M}(f, \Theta)(y) \in Z)}}\right| > \epsilon \right] \leq \gamma,
\end{equation}
where the outer probability is calculated with respect to the probability space $\mathrm{Range}(\mathcal{M} \circ \mathcal{G})$ obtained by applying the privacy-preserving mechanism $\mathcal{M}$ on the data-generation distribution $\mathcal{G}$.
\label{def:par}
\end{definition}

If a privacy-preserving mechanism is $\epsilon_0$-differentially private for a given query $f$ and parameters $\Theta$, for any privacy level $\epsilon \geq \epsilon_0$, privacy at risk is $0$. % Our interest is to study the effect of both the randomness induced by the noise and that of the data-generation distribution to help a data steward calibrate the privacy-preserving mechanism as per a desired privacy level and a desired confidence level.
Our interest is to quantify the risk $\gamma$ with which $\epsilon_0$-differentially private privacy-preserving mechanism also satisfies a stronger $\epsilon$-differential privacy, i.e. $\epsilon < \epsilon_0$.

\textbf{Unifying Probabilistic and Random Differential Privacy.} 
As a natural consequence, Equation~\ref{eqn:par} unifies the notions of probabilistic differential privacy and random differential privacy by accounting for both sources of randomness in a privacy-preserving mechanism.
\citeauthor{machanavajjhala2008privacy} define probabilistic differential privacy that incorporates the explicit randomness of the noise distribution of the privacy-preserving mechanism whereas \citeauthor{hall2012random} define random differential privacy that incorporates the implicit randomness of the data-generation distribution.
In probabilistic differential privacy, the outer probability is computed over the sample space of $\mathrm{Range}(\mathcal{M})$ and all datasets are equally probable.

% Thus, if we consider a uniform data-generation distribution in  Equation~\ref{eqn:par}, Definition~\ref{def:par} leads to the definition of probabilistic differential privacy. \cite{hall2012random} define random differential privacy that incorporates the implicit randomness of the data-generation distribution.
% In random differential privacy, the outer probability is calculated with respect to the support of the data-generation distribution $\mathcal{G}$.
% Thus, if we consider a specific data-generation distribution $\mathcal{G}$ in Equation~\ref{eqn:par}, Definition~\ref{def:par} leads to the definition of random differential privacy. 

% We do not only coalesce these two aspects but also extend them by providing analytical results connecting privacy level with risk for the Laplace mechanism.
% Now, we instantiate privacy at risk for the Laplace mechanism for three cases: two cases involving two sources of randomness and third case involving the coupled effect. Three different cases correspond to three different interpretations of the confidence level, represented by the parameter $\gamma$, corresponding to three interpretation of the support of the outer probability in Definition~\ref{def:par}. In order to highlight this nuance, we denote the confidence levels corresponding to the three cases and their three sources of randomness as $\gamma_1$, $\gamma_2$ and $\gamma_3$, respectively.

\subsection{Composition theorem}

Application of $\epsilon$-differential privacy to many real-world problem suffers from the degradation of privacy guarantee, i.e. privacy level, over the composition. The basic composition theorem~\citep{dwork2014algorithmic} dictates that the privacy guarantee degrades linear in the number of evaluations of the mechanism. Advanced composition theorem~\citep{dwork2014algorithmic} provides a finer analysis of the privacy loss over multiple evaluations and provides a square root dependence on the the number of evaluations. In this section, we provide the composition theorem for privacy at risk.

% Let $\{\mathcal{M}^{i}: \mathcal{D} \rightarrow R^{i}\}_{i=1}^n$ denote a set of $\epsilon_0$-differentially private mechanisms. We want to find a collective privacy guarantee provided application of $\{\mathcal{M}^{i}\}_{i=1}^n$ on a dataset $\mathcal{D}$. 
\begin{definition}{\textbf{[Privacy loss random variable.]}}
For a privacy-preserving mechanism $\mathcal{M}: \mathcal{D} \rightarrow R$ and two neighbouring datasets $x, y \in \mathcal{D}$, the privacy loss random variable $\mathcal{C}$ takes a value $r \in R$
\[
\mathcal{C} \triangleq  \ln{\frac{\mathbb{P}(\mathcal{M}(x))}{\mathbb{P}(\mathcal{M}(y))}}
\]
\end{definition}

\begin{lemma}
If a privacy-preserving mechanism $\mathcal{M}$ satisfies $\epsilon_0$ differential privacy, then
\[
\mathbb{P}[| \mathcal{C}| \leq \epsilon_0] = 1
\]
\end{lemma}

\begin{theorem}\label{thm:compose}
For all $\epsilon_0, \epsilon, \gamma, \delta > 0$, the class of $\epsilon_0$-differentially private mechanisms, which satisfy $(\epsilon, \gamma)$-privacy at risk, are $(\epsilon', \delta)$-differential privacy under $n$-fold composition where
\[
\epsilon' = \epsilon_0 \sqrt{2n\ln{\frac{1}{\delta}}} + n\mu
\]
where,
$\mu = [\gamma\epsilon(e^\epsilon - 1) + (1-\gamma)\epsilon_0(e^{\epsilon_0} - 1)]$
\end{theorem}
\begin{proof}
Let, $\mathcal{M}^{1...n}: \mathcal{D} \rightarrow R^{1} \times R^2 \times ... \times R^n$ denote the $n$-fold composition of privacy-preserving mechanisms  $\{\mathcal{M}^{i}: \mathcal{D} \rightarrow R^{i}\}_{i=1}^n$. Each $\epsilon_0$-differentially private $\mathcal{M}^i$ also satisfies $(\epsilon, \gamma)$-privacy at risk for some $\epsilon \leq \epsilon_0$ and appropriately computer $\gamma$. Consider any two neighbouring datasets $x, y \in \mathcal{D}$. Let,
\[
B = \{(r_1, ..., r_n) | \bigwedge_{i=1}^n \frac{\mathbb{P}(\mathcal{M}^i(x) = r_i)}{\mathbb{P}(\mathcal{M}^i(y) = r_i)} > \text{e}^\epsilon \}
\]
Using the technique in \citep[Theorem 3.20]{dwork2014algorithmic}, it suffices to show that $\mathbb{P}(\mathcal{M}^{1...n}(x) \in B) \leq \delta$.

Consider,
\begin{align}
&\ln{\frac{\mathbb{P}(\mathcal{M}^{1...n}(x) = (r_1, ..., r_n))}{\mathbb{P}(\mathcal{M}^{1...n}(y) = (r_1, ..., r_n))}}  \nonumber \\ 
=&\ln {\prod_{i=1}^n} \frac{\mathbb{P}(\mathcal{M}^i(x) = r_i)}{\mathbb{P}(\mathcal{M}^i(y) = r_i)} \nonumber \\
=& \sum_{i=1}^n \ln {\frac{\mathbb{P}(\mathcal{M}^i(x) = r_i)}{\mathbb{P}(\mathcal{M}^i(y) = r_i)}} ~=~ \sum_{i=1}^n \mathcal{C}^i
\end{align}
where $\mathcal{C}^i$ in the last line denotes privacy loss random variable related $\mathcal{M}^i$.

Consider, an $\epsilon$-differentially private mechanism $\mathcal{M}_\epsilon$ and $\epsilon_0$-differentially private mechanism $\mathcal{M}_{\epsilon_0}$. Let $\mathcal{M}_{\epsilon_0}$ satisfy $(\epsilon, \gamma)$-privacy at risk for $\epsilon \leq \epsilon_0$ and appropriately computed $\gamma$. Each $\mathcal{M}^i$ can be simulated as the mechanism $\mathcal{M}_\epsilon$ with probability $\gamma$ and the mechanism $\mathcal{M}_{\epsilon_0}$ otherwise. Therefore, privacy loss random variable for each mechanism $\mathcal{M}^i$ can be written as
\[
\mathcal{C}^i = \gamma \mathcal{C}_\epsilon^i + (1-\gamma) \mathcal{C}_{\epsilon_0}^i
\]
where, $\mathcal{C}_\epsilon^i$ denotes the privacy loss random variable associated with the mechanism $\mathcal{M}_\epsilon$ and $\mathcal{C}_{\epsilon_0}^i$ denotes the privacy loss random variable associated with the mechanism $\mathcal{M}_{\epsilon_0}$.
Using \citep[Lemma $3.18$]{dwork2014algorithmic}, we can bound the mean of every privacy loss random variable as,
\[
\mu \triangleq \mathbb{E}[\mathcal{C}^i] \leq [\gamma\epsilon(e^\epsilon - 1) + (1-\gamma)\epsilon_0(e^{\epsilon_0} - 1)]
\]
We have a collection of $n$ independent privacy random variables $\mathcal{C}^i$s such that $\mathbb{P}\left[|\mathcal{C}^i| \leq \epsilon_0 \right] = 1$. Using Hoeffding's bound~\citep{hoeffding1994probability} on the sample mean for any $\beta > 0$,
\[
\mathbb{P}\left[\frac{1}{n}\sum_i C^i \geq \mathbb{E}[\mathcal{C}^i] + \beta \right] \leq \exp{\left(-\frac{n\beta^2}{2\epsilon_0^2}\right)}
\]
Rearranging the inequality by renaming the upper bound on the probability as $\delta$, we get,
\[
\mathbb{P}\left[\sum_i C^i \geq n\mu + \epsilon_0 \sqrt{2n\ln{\frac{1}{\delta}}} \right] \leq \delta
\]
\end{proof}
% \emph{Theorem~\ref{thm:compose} improves the existing composition result~\citep[Theorem 3.20]{dwork2014algorithmic} by a factor of $2$.}

Theorem~\ref{thm:compose} is an analogue, in the privacy at risk setting, of the advanced composition of differential privacy~\citep[Theorem 3.20]{dwork2014algorithmic} under a constraint of independent evaluations. Note that, if one takes $\gamma = 0$, then we obtain the exact same formula as in~\citep[Theorem 3.20]{dwork2014algorithmic}. It provides a sanity check for the consistency of composition using privacy at risk.

In fact, if we consider both sources of randomness, the expected value of loss function must be computed by using the law of total expectation.
\[
\mathbb{E}[\mathcal{C}] = \mathbb{E}_{x,y \sim \mathcal{G}}[\mathbb{E}[\mathcal{C}] | x,y]
\]
Therefore, the exact computation of privacy guarantees after the composition requires access to the data-generation distribution. We assume a uniform data-generation distribution while proving Theorem~\ref{thm:compose}. We can obtain better and finer privacy guarantees accounting for data-generation distribution, which we keep as a future work.

% Additionally, using the privacy at risk improves the composition of the differential privacy further. We illustrate this in Figure~\ref{fig:composition}.
\subsection{Convexity and Post-processing}
Since privacy at risk provides a probabilistic privacy guarantee for an $\epsilon_0$-differentially private mechanism, it does not alter the basic properties of differential privacy - convexity and post-processing. We now show that privacy at risk equally adheres to both of these properties. 
% We stick to the formalism in~\cite{kifer2012axiomatic,desfontaines2019sok} for the following proofs.

\begin{lemma}[Convexity]
	For a given $\epsilon_0$-differentially private privacy-preserving mechanism, privacy at risk satisfies convexity property.
	\label{lemma:convexity}
\end{lemma}
\begin{proof}
	Let $\mathcal{M}$ be a mechanism that satisfies $\epsilon_0$-differential privacy. By the definition of the privacy at risk, it also satisfies $(\epsilon_1, \gamma_1)$-privacy at risk as well as $(\epsilon_2, \gamma_2)$-privacy at risk for some $\epsilon_1, \epsilon_2 \leq \epsilon_0$ and appropriately computed values of $\gamma_1$ and $\gamma_2$. Let $\mathcal{M}^1$ and $\mathcal{M}^2$ denote the hypothetical mechanisms that satisfy $(\epsilon_1, \gamma_1)$-privacy at risk and $(\epsilon_2, \gamma_2)$-privacy at risk respectively. We can write privacy loss random variables as follows:
	\begin{align}
	\mathcal{C}^1 &\leq \gamma_1 \epsilon_1 + (1 - \gamma_1)\epsilon_0 \nonumber \\
	\mathcal{C}^2 &\leq \gamma_2 \epsilon_2 + (1 - \gamma_2)\epsilon_0 \nonumber
	\end{align}
	where $\mathcal{C}^1$ and $\mathcal{C}^2$ denote privacy loss random variables for $\mathcal{M}^1$ and $\mathcal{M}^2$.
	
	Let us consider a privacy-preserving mechanism $\mathcal{M}$ that uses $\mathcal{M}^1$ with a probability $p$ and $\mathcal{M}^2$ with a probability $(1 - p)$ for some $p \in [0, 1]$. By using the techniques in the proof of Theorem~\ref{thm:compose}, the privacy loss random variable $\mathcal{C}$ for $\mathcal{M}$ can be written as:
	\begin{align}
	\mathcal{C} &= p\mathcal{C}^1 + (1 - p)\mathcal{C}^2 \nonumber \\
	&\leq \gamma'\epsilon' + (1 - \gamma')\epsilon_0 \nonumber
	\end{align}
	where 
	\begin{align}
	\epsilon' &= \frac{p\gamma_1\epsilon_1 + (1 - p)\gamma_2\epsilon_2}{p\gamma_1 + (1-p)\gamma_2} \nonumber \\
	\gamma' &= (1 - p\gamma_1 - (1-p)\gamma_2) \nonumber
	\end{align}
	Thus, $\mathcal{M}$ satisfies $(\epsilon', \gamma')$-privacy at risk. It proves that privacy at risk staisfies convexity~\cite[Axiom 2.1.2]{kifer2012axiomatic}.
\end{proof}

\begin{lemma}[Post-processing]
	For a given $\epsilon_0$-differentially private privacy-preserving mechanism, privacy at risk satisfies post-processing property.
	\label{lemma:post-processing}
\end{lemma}
\begin{proof}
	\[ \begin{tikzcd}
	(\epsilon_0, 0)-DP \arrow{r}{\varphi} \arrow[swap]{d}{} & (\epsilon_0, 0)-DP \arrow{d}{} \\%
	(\epsilon, \gamma)-PaR \arrow{r}{}& (\epsilon, \gamma)-PaR
	\end{tikzcd}
	\]
	Privacy at risk analyses sources of randomness involved in a privacy-preserving mechanism to provide probabilistic guarantees over the privacy level of differential privacy. Application of a deterministic function as a post-processing does not effect either data-generation distribution or noise-distribution. Thus privacy at risk simply reflects the features of underlying privacy level under post-processing. Since differential privacy is preserved under post-processing~\cite[Proposition 2.1]{dwork2014algorithmic}, so is the privacy at risk.
\end{proof}

We need to closely observe Lemma~\ref{lemma:post-processing}. Privacy at risk is preserved under post-processing only if the mechanism satisfies a stronger privacy guarantee such as differential privacy. If the mechanism satisfies a variant such probabilistic differential privacy, which does not satisfy post-processing property~\cite{machanavajjhala2008privacy}, then privacy at risk will not be preserved under post-processing.

\section{Privacy at Risk for Laplace Mechanism}
In this section, we instantiate privacy at risk for the Laplace mechanism for three cases: two cases involving two sources of randomness and third case involving the coupled effect. Three different cases correspond to three different interpretations of the confidence level, represented by the parameter $\gamma$, corresponding to three interpretation of the support of the outer probability in Definition~\ref{def:par}. In order to highlight this nuance, we denote the confidence levels corresponding to the three cases and their three sources of randomness as $\gamma_1$, $\gamma_2$ and $\gamma_3$, respectively.

\subsection{The Case of Explicit Randomness}
\label{sec:case_one}

In this section, we study the effect of the explicit randomness induced by the noise sampled from Laplacian distribution. We provide a probabilistic quantification for fine tuning for the Laplace mechanism. We fine-tune the privacy level for a specified risk under by assuming that the sensitivity of the query is known a priori.

For a Laplace mechanism $\mathcal{L}_{\epsilon_0}^{\Delta_f}$ calibrated with sensitivity $\Delta_f$ and privacy level $\epsilon_0$, we present the analytical formula relating privacy level $\epsilon$ and the risk $\gamma_1$ in Theorem~\ref{thm:par_1}. The proof is available in Appendix~\ref{app:case_1}.

\begin{theorem}
The risk $\gamma_1 \in [0,1]$ with which a Laplace Mechanism $\mathcal{L}^{\Delta_f}_{\epsilon_0}$, for a numeric query $f:\mathcal{D} \rightarrow \mathbb{R}^k$ satisfies a privacy level $\epsilon \geq 0$ is given by
\begin{equation}
\gamma_1 = \frac{\mathbb{P}(T \leq \epsilon)}{\mathbb{P}(T \leq \epsilon_0)},
\label{eqn:gamma_1}
\end{equation}
where $T$ is a random variable that follows a distribution with the following density function.
\[
P_{T}(t) = \frac{2^{1-k}t^{k- \frac{1}{2}}K_{k- \frac{1}{2}}(t)\epsilon_0}{\sqrt{2\pi} \Gamma(k) \Delta_f}
\]
where $K_{n- \frac{1}{2}}$ is the Bessel function of second kind.
\label{thm:par_1}
\end{theorem}

Figure~\ref{fig:general} shows the plot of the privacy level against risk for different values of $k$ and for a Laplace mechanism $\mathcal{L}^{1.0}_{1.0}$. As the value of $k$ increases, the amount of noise added in the output of numeric query increases. Therefore, for a specified privacy level, the privacy at risk level increases with the value of $k$.

The analytical formula representing $\gamma_1$ as a function of $\epsilon$ is bijective. We need to invert it to obtain the privacy level $\epsilon$ for a privacy at risk $\gamma_1$. However the analytical closed form for such an inverse function is not explicit. We use a numerical approach to compute privacy level for a given privacy at risk from the analytical formula of Theorem~\ref{thm:par_1}. 
% This corresponds to inverting the analytical formula for $\gamma_1$ by reading the relevant plot in Figure~\ref{fig:general} from the y-axis to the x-axis.

\begin{figure*}[!t]
\centering
\begin{minipage}{0.62\textwidth}
\begin{subfigure}{0.48\textwidth}
  \centering
  \includegraphics[width=\textwidth]{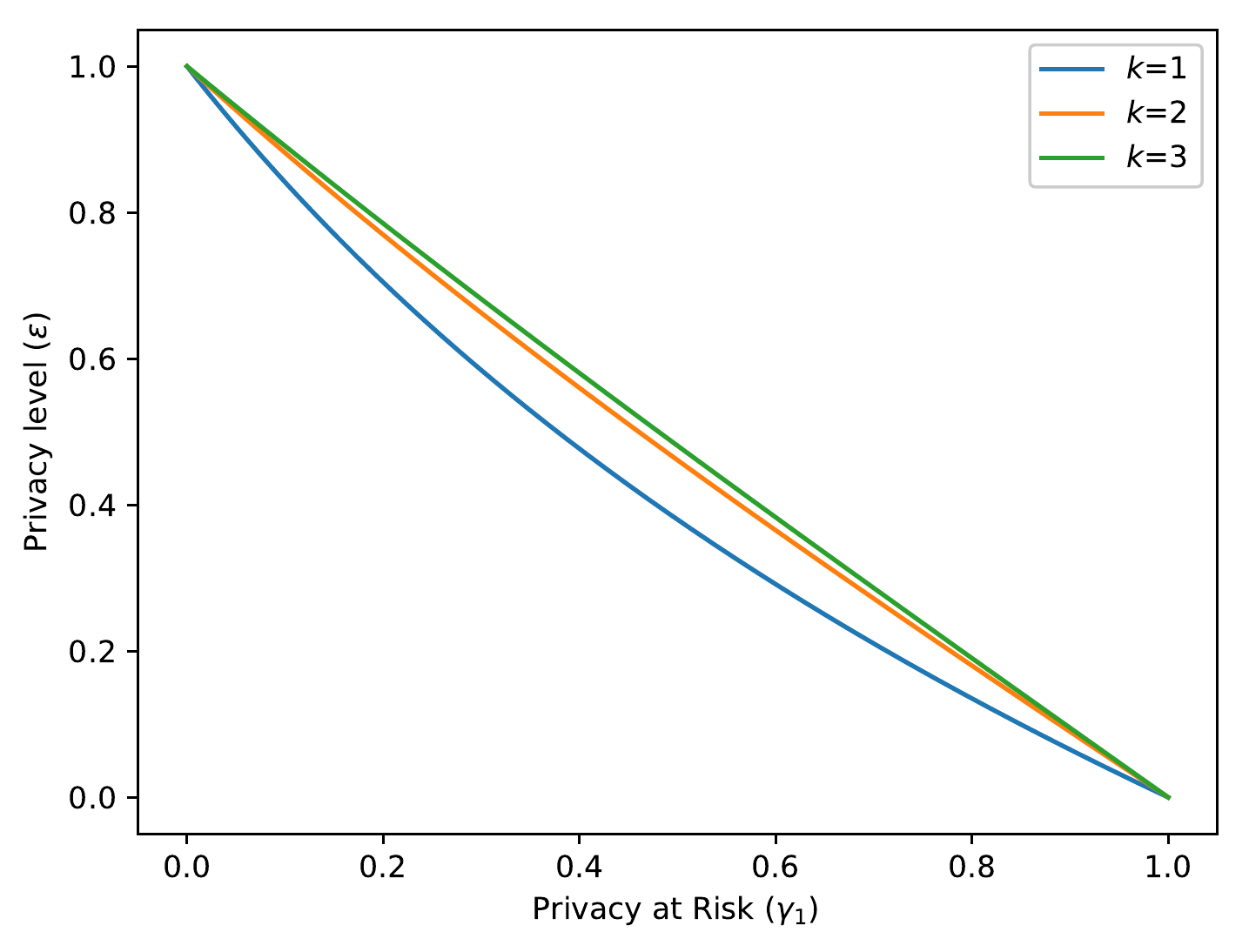}
  \captionof{figure}{}
  \label{fig:general}
\end{subfigure}%
\hspace*{1em}
\begin{subfigure}{0.48\textwidth}
  \centering
  \includegraphics[width=\textwidth]{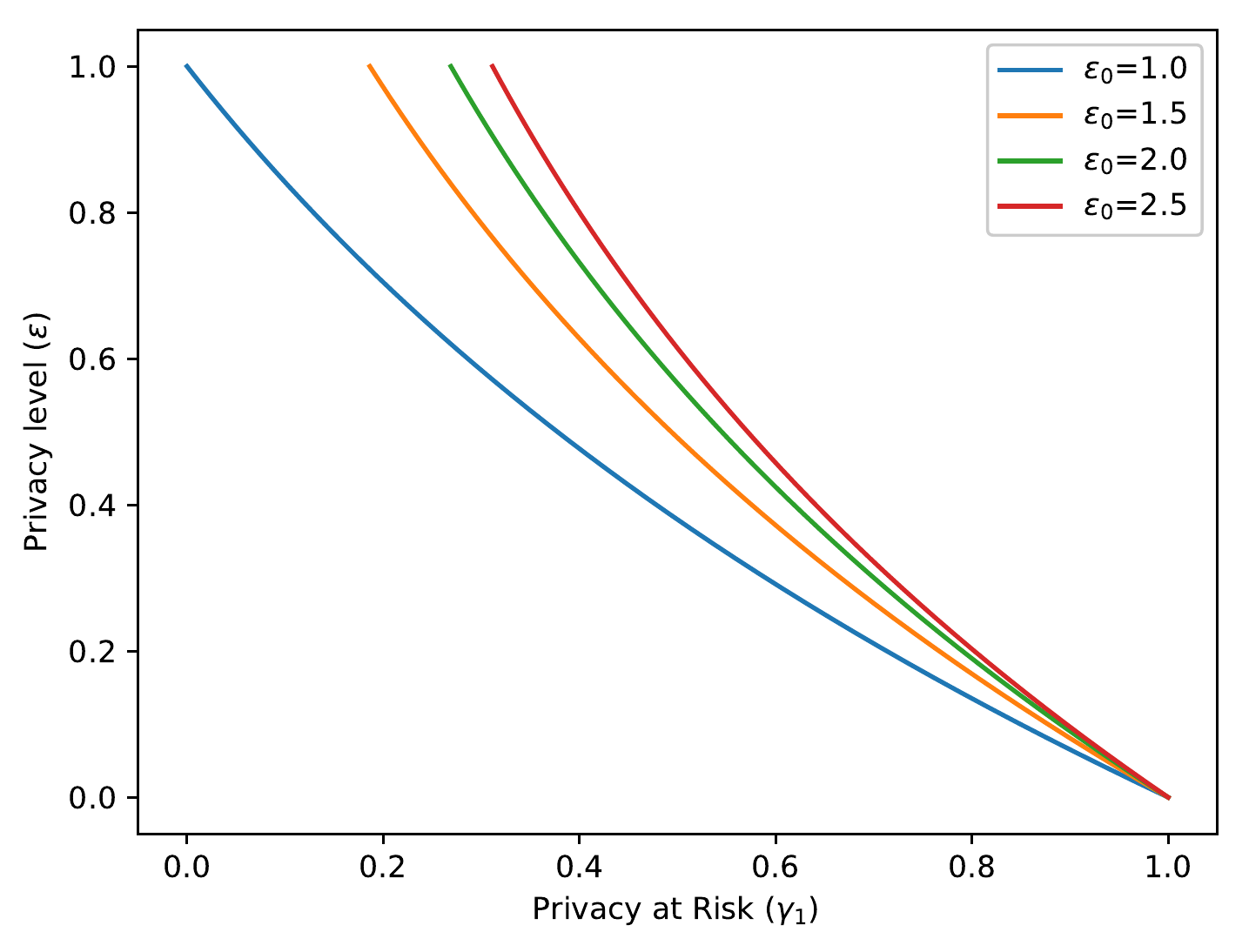}
  \captionof{figure}{}
  \label{fig:single}
\end{subfigure}\vspace*{-1em}
\caption{Privacy level $\epsilon$ for varying privacy at risk $\gamma_1$ for Laplace mechanism $\mathcal{L}_{\epsilon_0}^{1.0}$. In Figure~\ref{fig:general}, we use $\epsilon_0 = 1.0$ and different values of $k$. In Figure~\ref{fig:single}, for $k=1$ and different values of $\epsilon_0$.}
%\label{fig:par_3}
\end{minipage}
\hspace*{0.5em}
\begin{minipage}{0.35\textwidth}
\includegraphics[width=0.9\textwidth]{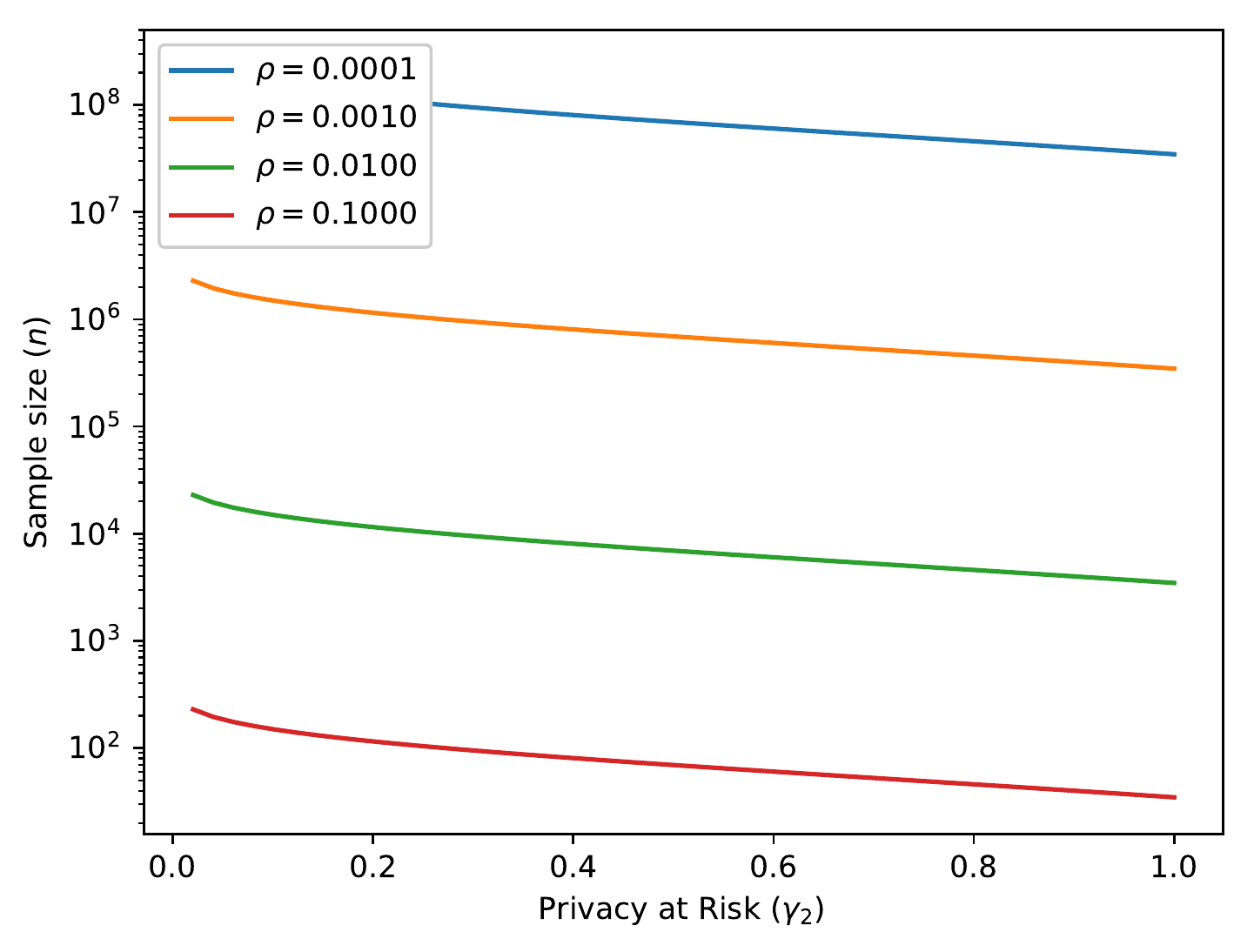}
  \caption{Number of samples $n$ for varying privacy at risk $\gamma_2$ for different error parameter $\rho$.}
  \label{fig:sample_size}
\end{minipage}\vspace{-1.2em}
\end{figure*}

\textbf{Result for a Real-valued Query.} For the case $k=1$, the analytical derivation is fairly straightforward. In this case, we obtain an invertible closed-form of a privacy level for a specified risk. It is presented in Equation~\ref{eq:special_case}.
\begin{equation}
    \epsilon = \ln{\left(\frac{1}{1 - \gamma_1(1 - e^{-\epsilon_0})}\right)}
    \label{eq:special_case}
\end{equation}

\textbf{Remarks on $\epsilon_0$.} For $k=1$, Figure~\ref{fig:single} shows the plot of  privacy at risk level $\epsilon$ versus privacy at risk $\gamma_1$ for the Laplace mechanism $\mathcal{L}^{1.0}_{\epsilon_0}$. As the value of $\epsilon_0$ increases, the probability of Laplace mechanism generating higher value of noise reduces. Therefore, for a fixed privacy level, privacy at risk increases with the value of $\epsilon_0$. The same observation is made for $k > 1$.

\subsection{The Case of Implicit Randomness}
\label{sec:case_two}

In this section, we study the effect of the implicit randomness induced by the data-generation distribution to provide a fine tuning for the Laplace mechanism. We fine-tune the risk for a specified privacy level without assuming that the sensitivity of the query.

If one takes into account randomness induced by the data-generation distribution, all pairs of neighbouring datasets are not equally probable. This leads to estimation of sensitivity of a query for a specified data-generation distribution. If we have access to an analytical form of the data-generation distribution and to the query, we could analytically derive the sensitivity distribution for the query. In general, we have access to the datasets, but not the data-generation distribution that generates them. We, therefore, statistically estimate sensitivity by constructing an empirical distribution. We call the sensitivity value obtained for a specified risk from the empirical cumulative distribution of sensitivity the \textit{sampled sensitivity} (Definition~\ref{def:sampled_sensitivity}). However, the value of sampled sensitivity is simply an estimate of the sensitivity for a specified risk. In order to capture this additional uncertainty introduced by the estimation from the empirical sensitivity distribution rather than the true unknown distribution, we compute a lower bound on the accuracy of this estimation. This lower bound yields a probabilistic lower bound on the specified risk. We refer to it as \textit{empirical risk}. For a specified absolute risk $\gamma_2$, we denote by $\hat{\gamma_2}$ corresponding empirical risk.

For the Laplace mechanism $\mathcal{L}_{\epsilon}^{\Delta_{S_f}}$ calibrated with sampled sensitivity $\Delta_{S_f}$ and privacy level $\epsilon$, we evaluate the empirical risk $\hat{\gamma_2}$. We present the result in Theorem~\ref{thm:par_2}. The proof is available in Appendix~\ref{app:case_2}.

\begin{theorem}
Analytical bound on the empirical risk, $\hat{\gamma_2}$, for Laplace mechanism $\mathcal{L}_{\epsilon}^{\Delta_{S_f}}$ with privacy level $\epsilon$ and sampled sensitivity $\Delta_{S_f}$ for a query $f:\mathcal{D} \rightarrow \mathbb{R}^k$ is
\begin{equation}
\hat{\gamma_2} \geq \gamma_2 (1 - 2e^{-2\rho^2n})
\label{eqn:gamma_2}
\end{equation}
where $n$ is the number of samples used for estimation of the sampled sensitivity and $\rho$ is the accuracy parameter. $\gamma_2$ denotes the specified absolute risk.
\label{thm:par_2}
\end{theorem}

The error parameter $\rho$ controls the closeness between the empirical cumulative distribution of the sensitivity to the true cumulative distribution of the sensitivity. Lower the value of the error, closer is the empirical cumulative distribution to the true cumulative distribution. Figure~\ref{fig:sample_size} shows the plot of number of samples as a function of the privacy at risk and the error parameter. Naturally, we require higher number of samples in order to have lower error rate. The number of samples reduces as the privacy at risk increases. The lower risk demands precision in the estimated sampled sensitivity, which in turn requires larger number of samples.

% If the analytical form of the data-generation distribution is not known a priori, the empirical distribution of sensitivity can be estimated in two ways. The first way is to fit a known distribution on the available data and later use it to build an empirical distribution of the sensitivities. The second way is to sub-sample from a large dataset in order to build an empirical distribution of the sensitivities. In both of these ways, the empirical distribution of sensitivities captures the inherent randomness in the data-generation distribution. The first way suffers from the goodness of the fit of the known distribution to the available data. An ill-fit distribution does not reflect the true data-generation distribution and hence introduces errors in the sensitivity estimation. Since the second way involves subsampling, it is immune to this problem. The quality of sensitivity estimates obtained by sub-sampling the datasets depend on the availability of large population to sample from.

Let, $\mathcal{G}$ denotes the data-generation distribution, either known apriori or constructed by subsampling the available data. We adopt the procedure of~\citep{rubinstein2017pain} to sample two neighbouring datasets with $p$ data points each. We sample $p - 1$ data points from $\mathcal{G}$ that are common to both of these datasets and later two more data points. From those two points, we allot one data point to each of the two datasets. 

Let, $S_f = \lVert f(x) - f(y) \rVert_1$ denotes the sensitivity random variable for a given query $f$, where $x$ and $y$ are two neighbouring datasets sampled from $\mathcal{G}$. Using $n$ pairs of neighbouring datasets sampled from $\mathcal{G}$, we construct the empirical cumulative distribution, $F_n$, for the sensitivity random variable.

\begin{definition}
For a given query $f$ and for a specified risk $\gamma_2$, sampled sensitivity, $\Delta_{S_f}$, is defined as the value of sensitivity random variable that is estimated using its empirical cumulative distribution function, $F_n$, constructed using $n$ pairs of neighbouring datasets sampled from the data-generation distribution $\mathcal{G}$. 
\[
\Delta_{S_f} \triangleq F_n^{-1}(\gamma_2)
\]
\label{def:sampled_sensitivity}
% \vspace*{-1.5em}
\end{definition}

If we knew analytical form of the data generation distribution, we could analytically derive the cumulative distribution function of the sensitivity, $F$, and find the sensitivity of the query as $\Delta_f = F^{-1}(1)$. Therefore, in order to have the sampled sensitivity close to the sensitivity of the query, we require the empirical cumulative distributions to be close to the cumulative distribution of the sensitivity. We use this insight to derive the analytical bound in the Theorem~\ref{thm:par_2}.

% \begin{figure}[t]
% \centering
%   \includegraphics[width=\linewidth]{images/sample_size}
%   \caption{Number of samples $n$ for varying confidence level of $\gamma_2$ for different error parameter $\rho$.}
%   \label{fig:sample_size}
% \end{figure}

\begin{figure*}[t]
\centering
\begin{subfigure}{0.5\linewidth}
  \centering
  \includegraphics[width=\textwidth]{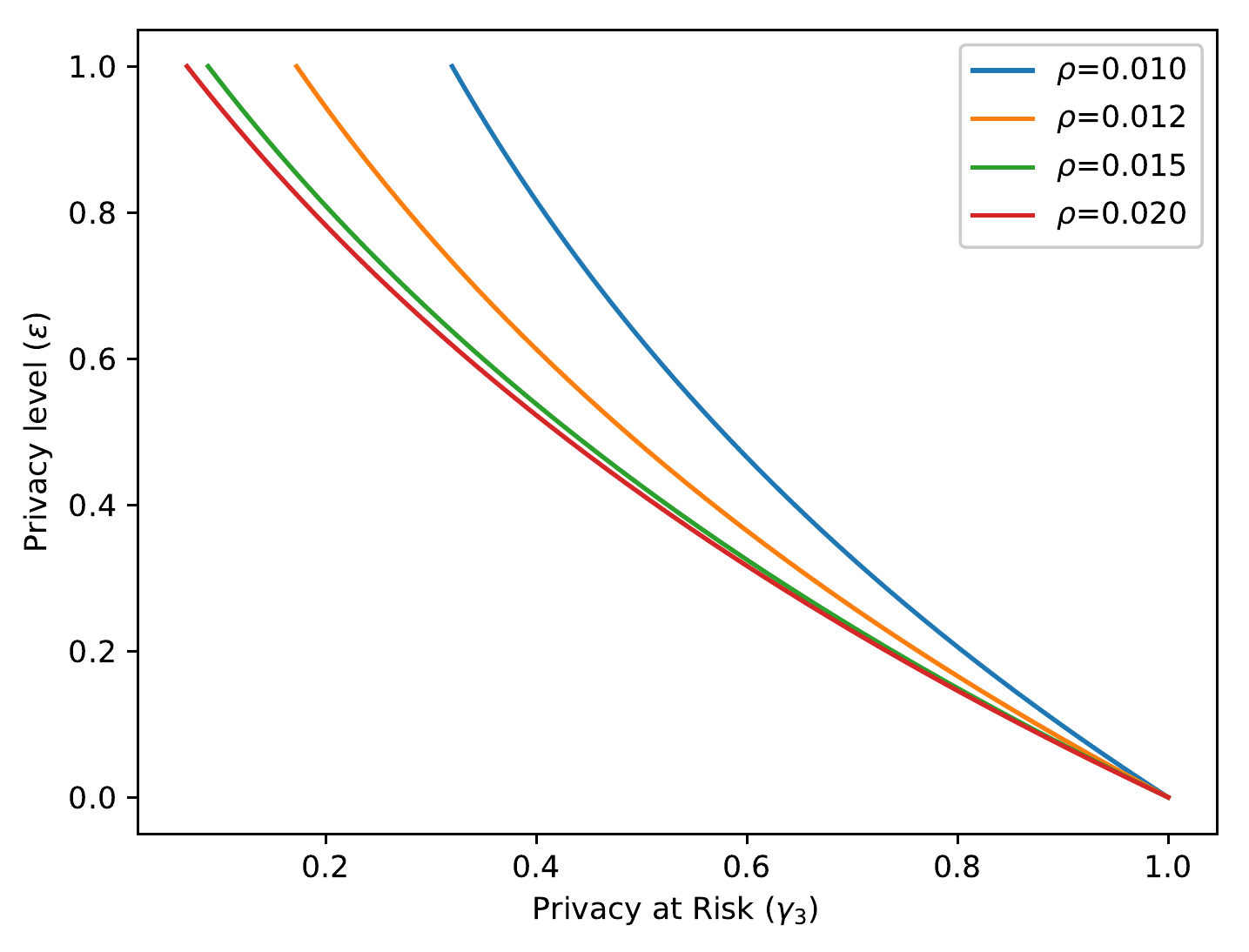}
%   \captionof{figure}{Privacy at risk level $\epsilon$ for varying confidence levels $\gamma_3$ for different error parameters $\rho$. We fix the number of samples to $10000$.}
  \vspace{-1.5em}  \captionof{figure}{}
  \label{fig:rho_dep}
\end{subfigure}%
\hspace*{1em}
\begin{subfigure}{0.5\linewidth}
  \centering
  \includegraphics[width=\textwidth]{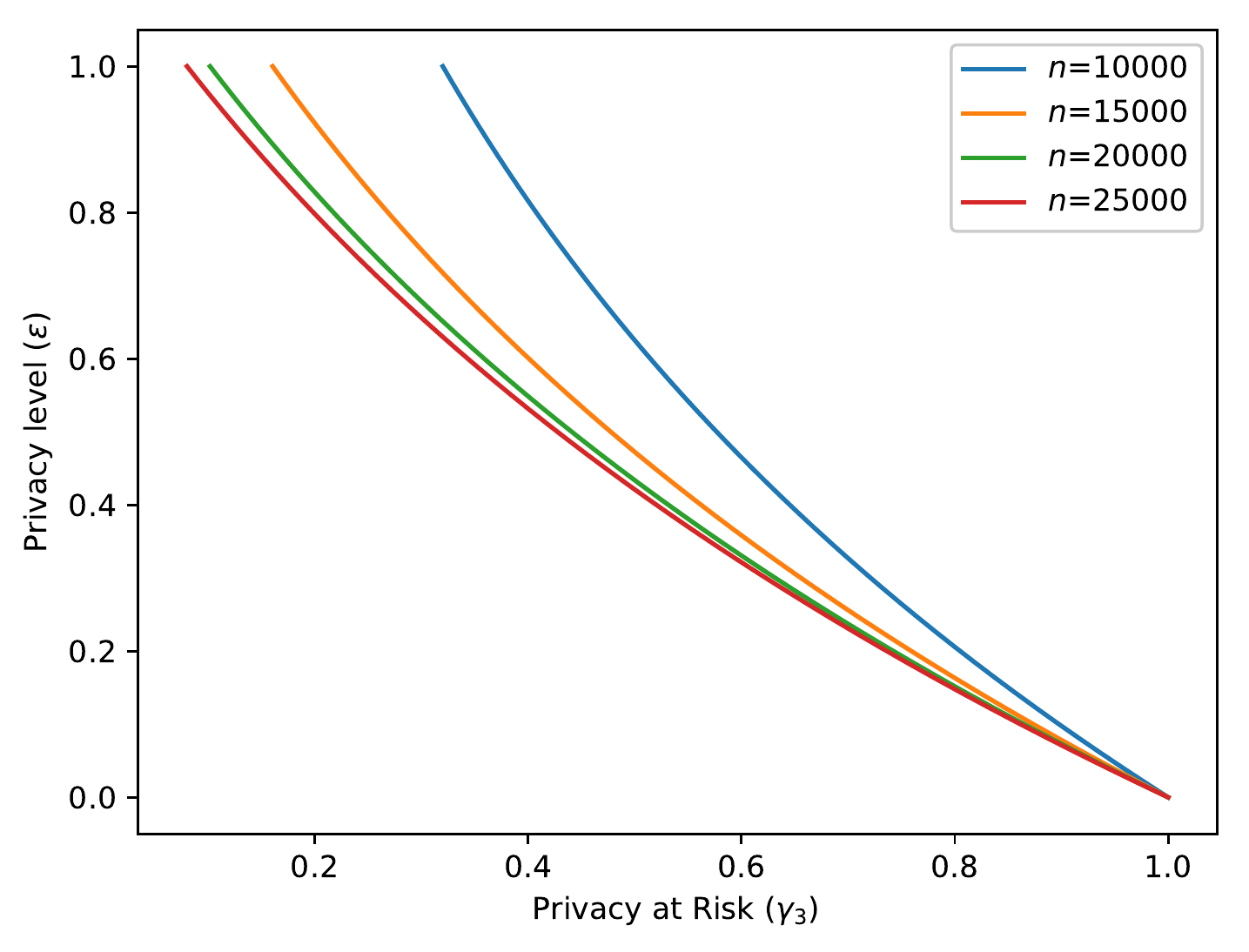}
%   \captionof{figure}{Privacy at risk level $\epsilon$ for varying confidence levels $\gamma_3$ for different sample sizes $n$. We fix the error parameter to $0.01$}
   \vspace{-1.5em} \captionof{figure}{}
  \label{fig:sample_dep}
\end{subfigure}\vspace{-1.2em}
\caption{Dependence of error and number of samples on the privacy at risk for Laplace mechanism $\mathcal{L}_{1.0}^{\Delta_{S_f}}$. For the figure on the left hand side, we fix the number of samples to $10000$. For the Figure~\ref{fig:sample_dep} we fix the error parameter to $0.01$.}
\label{fig:par_3}
\end{figure*}

\subsection{The Case of Explicit and Implicit Randomness}
\label{sec:case_three}

In this section, we study the combined effect of both explicit randomness induced by the noise distribution and implicit randomness in the data-generation distribution respectively. We do not assume the knowledge of the sensitivity of the query. 
% We study privacy level at a given risk $\gamma_3$ on the joint support of the noise distribution and the data-generation distribution.

We estimate sensitivity using the empirical cumulative distribution of sensitivity. We construct the empirical distribution over the sensitivities using the sampling technique presented in the earlier case. Since we use the sampled sensitivity (Definition~\ref{def:sampled_sensitivity}) to calibrate the Laplace mechanism, we estimate the \emph{empirical risk} $\hat{\gamma_3}$.

For Laplace mechanism $\mathcal{L}_{\epsilon_0}^{\Delta_{S_f}}$ calibrated with sampled sensitivity $\Delta_{S_f}$ and privacy level $\epsilon_0$, we present the analytical bound on the empirical sensitivity $\hat{\gamma_3}$ in Theorem~\ref{thm:par_3} with proof in the Appendix~\ref{app:case_3}.

\begin{theorem}
	Analytical bound on the empirical risk $\hat{\gamma_3} \in [0, 1]$ to achieve a privacy level $\epsilon > 0$ for Laplace mechanism $\mathcal{L}_{\epsilon_0}^{\Delta_{S_f}}$ with sampled sensitivity $\Delta_{S_f}$ of a query $f: \mathcal{D} \rightarrow \mathbb{R}^k$ is
	\begin{equation}
	\hat{\gamma_3} \geq \gamma_3 (1 - 2e^{-2\rho^2n})
	\end{equation}
	where $n$ is the number of samples used for estimating the sensitivity, $\rho$ is the accuracy parameter. $\gamma_3$ denotes the specified absolute risk defined as:
	\[
	\gamma_3 = \frac{\mathbb{P}(T \leq \epsilon)}{\mathbb{P}(T \leq \eta\epsilon_0)} \cdot \gamma_2
	\]
	\label{thm:par_3}
\end{theorem}
The error parameter $\rho$ controls the closeness between the empirical cumulative distribution of the sensitivity to the true cumulative distribution of the sensitivity. Figure~\ref{fig:par_3} shows the dependence of the error parameter on the number of samples. In Figure~\ref{fig:rho_dep}, we observe that the for a fixed number of samples and a privacy level, the privacy at risk decreases with the value of error parameter. For a fixed number of samples, smaller values of the error parameter reduce the probability of similarity between the empirical cumulative distribution of sensitivity and the true cumulative distribution. Therefore, we observe the reduction in the risk for a fixed privacy level. In Figure~\ref{fig:sample_dep}, we observe that for a fixed value of error parameter and a fixed level of privacy level, the risk increases with the number of samples. For a fixed value of the error parameter, larger values of the sample size increase the probability of similarity between the empirical cumulative distribution of sensitivity and the true cumulative distribution. Therefore, we observe the increase in the risk for a fixed privacy level.

Effect of the consideration of implicit and explicit randomness is evident in the analytical expression for $\gamma_3$ in Equation~\ref{eqn:gamma_3_repeated}. Proof is available in Appendix~\ref{app:case_3}. The privacy at risk is composed of two factors whereas the second term is a privacy at risk that accounts for inherent randomness. The first term takes into account the implicit randomness of the Laplace distribution along with a coupling coefficient $\eta$. We define $\eta$ as the ratio of the true sensitivity of the query to its sampled sensitivity.
\begin{equation}
\gamma_3 \triangleq \frac{\mathbb{P}(T \leq \epsilon)}{\mathbb{P}(T \leq \eta\epsilon_0)} \cdot \gamma_2
\label{eqn:gamma_3_repeated}
\end{equation}

% \section{Composition Theorem for Privacy at Risk}

\section{Minimising Compensation Budget for Privacy at Risk}
\label{sec:appl2}
Many service providers collect users' data to enhance user experience. In order to avoid misuse of this data, we require a legal framework that not only limits the use of the collected data but also proposes reparative measures in case of a data leak. General Data Protection Regulation (GDPR)\footnote{\url{https://eugdpr.org/}} is such a legal framework.
% Starting from May 2018, every business entity that holds or processes data of EU citizens, irrespective of the geographical location of the entity itself, must comply with GDPR. 

Section 82 in GDPR states that any person who suffers from material or non-material damage as a result of a personal data breach has the right to demand compensation from the data processor. Therefore, every GDPR compliant business entity that either holds or processes personal data needs to secure a certain budget in the worst case scenario of the personal data breach. In order to reduce the risk of such an unfortunate event, the business entity may use privacy-preserving mechanisms that provide provable privacy guarantees while publishing their results. In order to calculate the compensation budget for a business entity, we devise a cost model that maps the privacy guarantees provided by differential privacy and privacy at risk to monetary costs. The discussions demonstrate the usefulness of probabilistic quantification of differential privacy in a business setting.

\subsection{Cost Model for Differential Privacy}
Let $E$ be the compensation budget that a business entity has to pay to every stakeholder in case of a personal data breach when the data is processed without any provable privacy guarantees. Let $E_\epsilon^{dp}$ be the compensation budget that a business entity has to pay to every stakeholder in case of a personal data breach when the data is processed with privacy guarantees in terms of $\epsilon$-differential privacy.

Privacy level, $\epsilon$, in $\epsilon$-differential privacy is the quantifier of indistinguishability of the outputs of a privacy-preserving mechanism when two neighbouring datasets are provided as inputs. When the privacy level is zero, the privacy-preserving mechanism outputs all results with equal probability. The indistinguishability reduces with increase in the privacy level. Thus, privacy level of zero bears the lowest risk of personal data breach and the risk increases with the privacy level. $E_\epsilon^{dp}$ needs to be commensurate to such a risk and, therefore, it needs to satisfy the following constraints.
\begin{enumerate}
	\item For all $\epsilon \in \mathbb{R}^{\geq 0}$, $E_\epsilon^{dp} \leq E$.
	\item $E_\epsilon^{dp}$ is a monotonically increasing function of $\epsilon$.
	\item As $\epsilon \rightarrow 0$, $E_\epsilon^{dp} \rightarrow E_{min}$ where $E_{min}$ is the unavoidable cost that business entity might need to pay in case of personal data breach even after the privacy measures are employed.
	\item As $\epsilon \rightarrow \infty$, $E_\epsilon^{dp} \rightarrow E$.
\end{enumerate}
There are various functions that satisfy these constraints. For instantiation, we choose to work with a cost model that is convex with respect to $\epsilon$. The cost model $E_\epsilon^{dp}$ as defined in Equation~\ref{eq:dp_cost_model}. 
\begin{equation}
E_\epsilon^{dp} \triangleq E_{min} + Ee^{-\frac{c}{\epsilon}}
\label{eq:dp_cost_model}
\end{equation}
$E_\epsilon^{dp}$ has two parameters, namely $c > 0$ and $E_{min} \geq 0$. 
$c$ controls the rate of change in the cost as the privacy level changes and $E_{min}$ is a privacy level independent bias. For this study, we use a simplified model with $c = 1$ and $E_{min} = 0$.

\subsection{Cost Model for Privacy at Risk}
Let $E_{\epsilon_0}^{par}(\epsilon, \gamma)$ be the compensation that a business entity has to pay to every stakeholder in case of a personal data breach when the data is processed with an $\epsilon_0$-differentially private privacy-preserving mechanism along with a probabilistic quantification of privacy level. Use of such a quantification allows use to provide a stronger a stronger privacy guarantee \textit{viz.} $\epsilon < \epsilon_0$ for a specified privacy at risk at most $\gamma$ for  Thus, we calculate $E_{\epsilon_0}^{par}$ using Equation~\ref{eq:par_cost_model}.
\begin{equation}
E_{\epsilon_0}^{par}(\epsilon, \gamma) \triangleq \gamma E_\epsilon^{dp} + (1 - \gamma) E_{\epsilon_0}^{dp}
\label{eq:par_cost_model}
\end{equation}

\subsubsection{Existence of Minimum Compensation Budget} 
We want to find the privacy level, say $\epsilon_{min}$, that yields the lowest compensation budget.  We do that by minimising Equation~\ref{eq:par_cost_model} with respect to $\epsilon$.

\begin{lemma}\label{lemma:compensation}
$E_{\epsilon_0}^{par}(\epsilon, \gamma)$ is a convex function of $\epsilon$ if $E_\epsilon^{dp}$ is defined by Equation~\eqref{eq:par_cost_model}.
\end{lemma}
This result also generalises to any other convex cost model satisfying the four conditions.
By Lemma~\ref{lemma:compensation}, there exists a unique $\epsilon_{min}$ that minimises the compensation budget for a specified parametrisation, say $\epsilon_0$. Since the risk $\gamma$ in Equation~\ref{eq:par_cost_model} is itself a function of privacy level $\epsilon$, analytical calculation of $\epsilon_{min}$ is not possible in the most general case. When the output of the query is a real number, we derive the analytic form (Equation~\ref{eq:special_case}) to compute the risk under the consideration of explicit randomness. In such a case, $\epsilon_{min}$ is calculated by differentiating Equation~\ref{eq:par_cost_model} with respect to $\epsilon$ and equating it to zero. It gives us Equation~\ref{eq:e_min} that we solve using any root finding technique such as Newton-Raphson method~\citep{press2007numerical} to compute $\epsilon_{min}$.

\begin{equation}
    \frac{1}{\epsilon} - \ln{\left(1 - \frac{1 - e^\epsilon}{\epsilon^2}\right)} = \frac{1}{\epsilon_0}
    \label{eq:e_min}
\end{equation}

\subsubsection{Fine-tuning Privacy at Risk} 
For a fixed budget, say $B$, re-arrangement of Equation~\ref{eq:par_cost_model} gives us an upper bound on the privacy level $\epsilon$. 
We use the cost model with $c = 1$ and $E_{min} = 0$ to derive the upper bound. 
If we have a maximum permissible expected mean absolute error $T$, we use Equation~\ref{eq:mae} to obtain a lower bound on the privacy at risk level. Equation~\ref{eq:bound_par} illustrates the upper and lower bounds that dictate the permissible range of $\epsilon$ that a data publisher can promise depending on the budget and the permissible error constraints.
\begin{equation}
	\frac{1}{T} \leq \epsilon \leq \left[\ln{\left(\frac{\gamma E}{B - (1 - \gamma)E_{\epsilon_0}^{dp}}\right)}\right]^{-1}
	\label{eq:bound_par}
\end{equation}

Thus, the privacy level is constrained by the effectiveness requirement from below and by the monetary budget from above. \citep{hsu2014differential} calculate upper and lower bound on the privacy level in the differential privacy. They use a different cost model owing to the scenario of research study that compensates its participants for their data and releases the results in a differentially private manner. Their cost model is different than our GDPR inspired modelling.

\subsection{Illustration}
\label{sec:cost_illustration}
Suppose that the health centre in a university that complies to GDPR publishes statistics of its staff health checkup, such as obesity statistics, twice in a year. In January 2018, the health centre publishes that 34 out of 99 faculty members suffer from obesity. In July 2018, the health centre publishes that 35 out of 100 faculty members suffer from obesity. An intruder, perhaps an analyst working for an insurance company, checks the staff listings in January 2018 and July 2018, which are publicly available on website of the university. The intruder does not find any change other than the recruitment of John Doe in April 2018. Thus, with high probability, the intruder deduces that John Doe suffers from obesity. In order to avoid such a privacy breach, the health centre decides to publish the results using the Laplace mechanism. In this case, the Laplace mechanism operates on the count query.

\begin{figure}[!t]
	\centering
	\includegraphics[width=\linewidth]{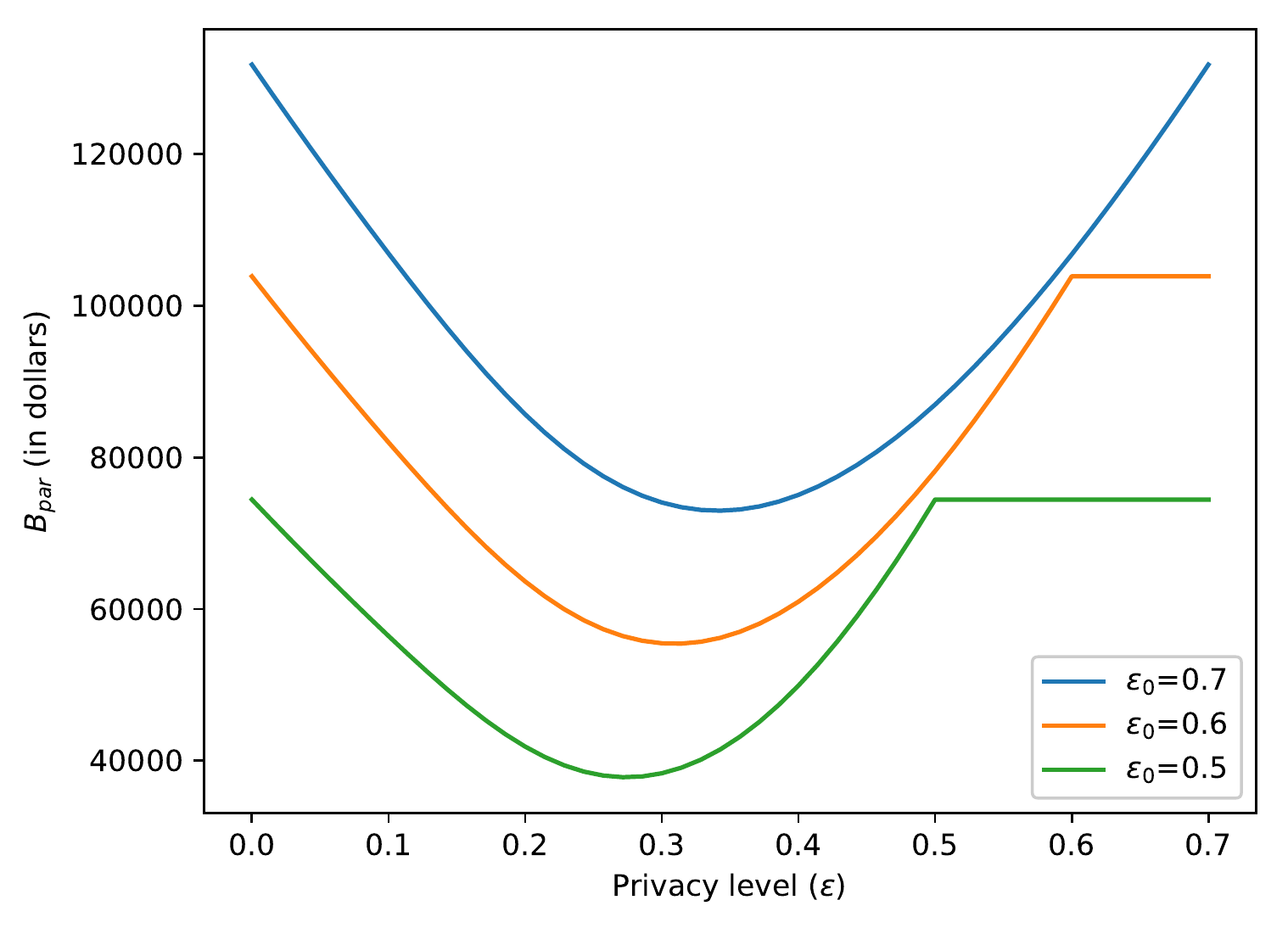}
	\caption{Variation in the budget for Laplace mechanism $\mathcal{L}_{\epsilon_0}^1$ under privacy at risk considering explicit randomness in the Laplace mechanism for the illustration in Section~\ref{sec:cost_illustration}.}
	\label{fig:b_par}\vspace{-1.2em}
\end{figure}
In order to control the amount of noise, the health centre needs to appropriately set the privacy level. Suppose that the health centre decides to use the expected mean absolute error, defined in Equation~\ref{eq:mae}, as the measure of \textit{effectiveness} for the Laplace mechanism.
\begin{equation}
\mathbb{E}\left[\lvert \mathcal{L}^{1}_\epsilon(x) - f(x) \rvert \right] = \frac{1}{\epsilon}
\label{eq:mae}
\end{equation}
Equation~\ref{eq:mae} makes use of the fact that the sensitivity of the count query is one. Suppose that the health centre requires the expected mean absolute error of at most two in order to maintain the quality of the published statistics. In this case, the privacy level has to be at least $0.5$. 

In order to compute the budget, the health centre requires an estimate of $E$. \cite{moriarty2012effects} shows that the incremental cost of premiums for the health insurance with morbid obesity ranges between $\$5467$ to $\$5530$.
With reference to this research, the health centre takes $\$5500$ as an estimate of $E$. For the staff size of $100$ and the privacy level $0.5$, the health centre uses Equation~\ref{eq:dp_cost_model} in its simplified setting to compute the total budget of $\$74434.40$.

Is it possible to reduce this budget without degrading the effectiveness of the Laplace mechanism? We show that it is possible by fine-tuning the Laplace mechanism. Under the consideration of the explicit randomness introduced by the Laplace noise distribution, we show that $\epsilon_0$-differentially private Laplace mechanism also satisfies $\epsilon$-differential privacy with risk $\gamma$, which is computed using the formula in Theorem~\ref{thm:par_1}. Fine-tuning allows us to get a stronger privacy guarantee, $\epsilon < \epsilon_0$ that requires a smaller budget. In Figure~\ref{fig:b_par}, we plot the budget for various privacy levels. We observe that the privacy level $0.274$, which is same as $\epsilon_{min}$ computed by solving Equation~\ref{eq:e_min}, yields the lowest compensation budget of $\$37805.86$. Thus, by using privacy at risk, the health centre is able to save $\$36628.532$ without sacrificing the quality of the published results.

\begin{figure*}[!t]
\centering
\begin{subfigure}{0.33\textwidth}
  \centering
  \includegraphics[width=\textwidth]{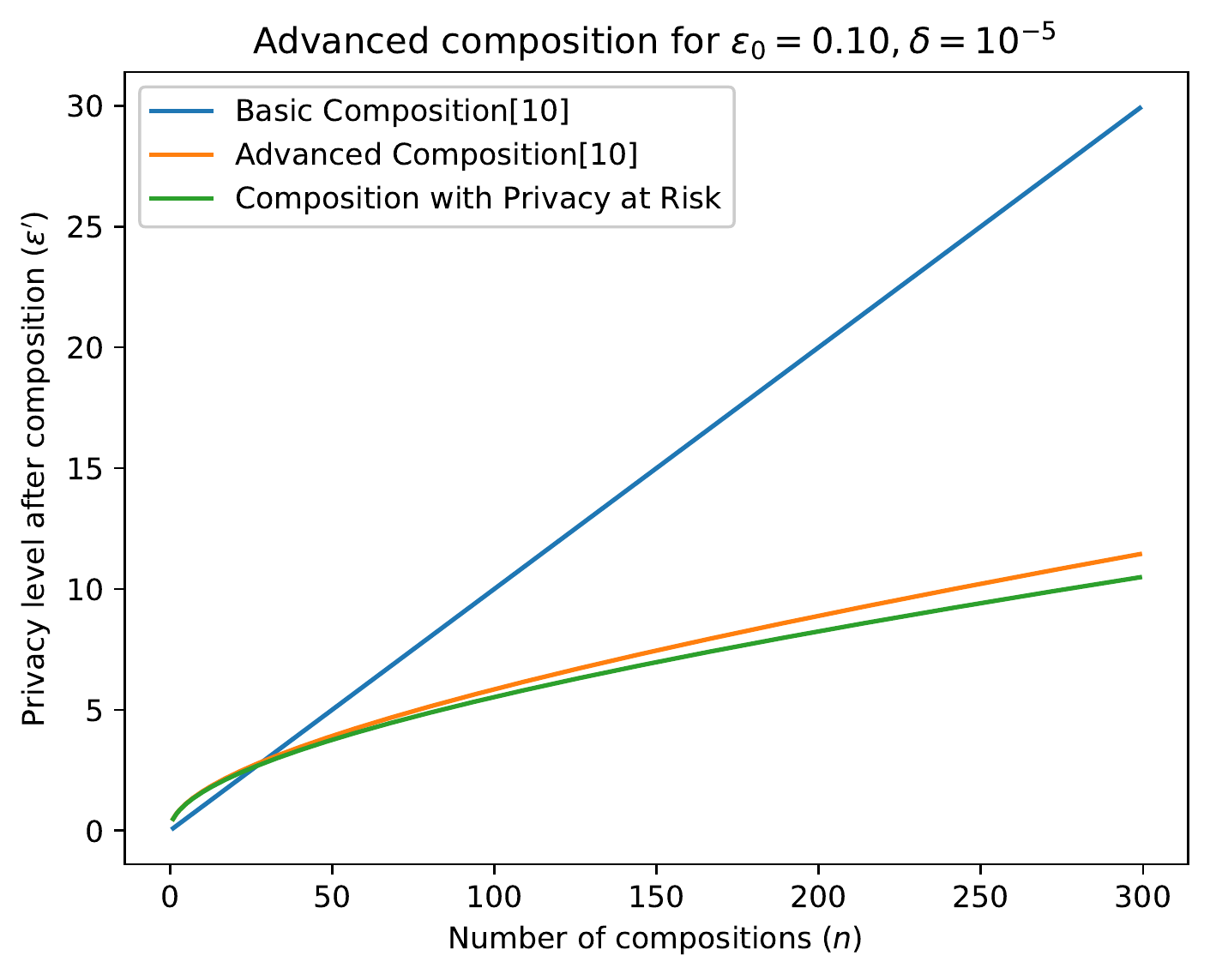}
  \captionof{figure}{$\mathcal{L}_{0.1}^1$ satisfies $(0.08,0.80)$-privacy at risk.}
  \label{fig:composition_0}
\end{subfigure}%
\hspace*{1em}
\begin{subfigure}{0.33\textwidth}
  \centering
  \includegraphics[width=\textwidth]{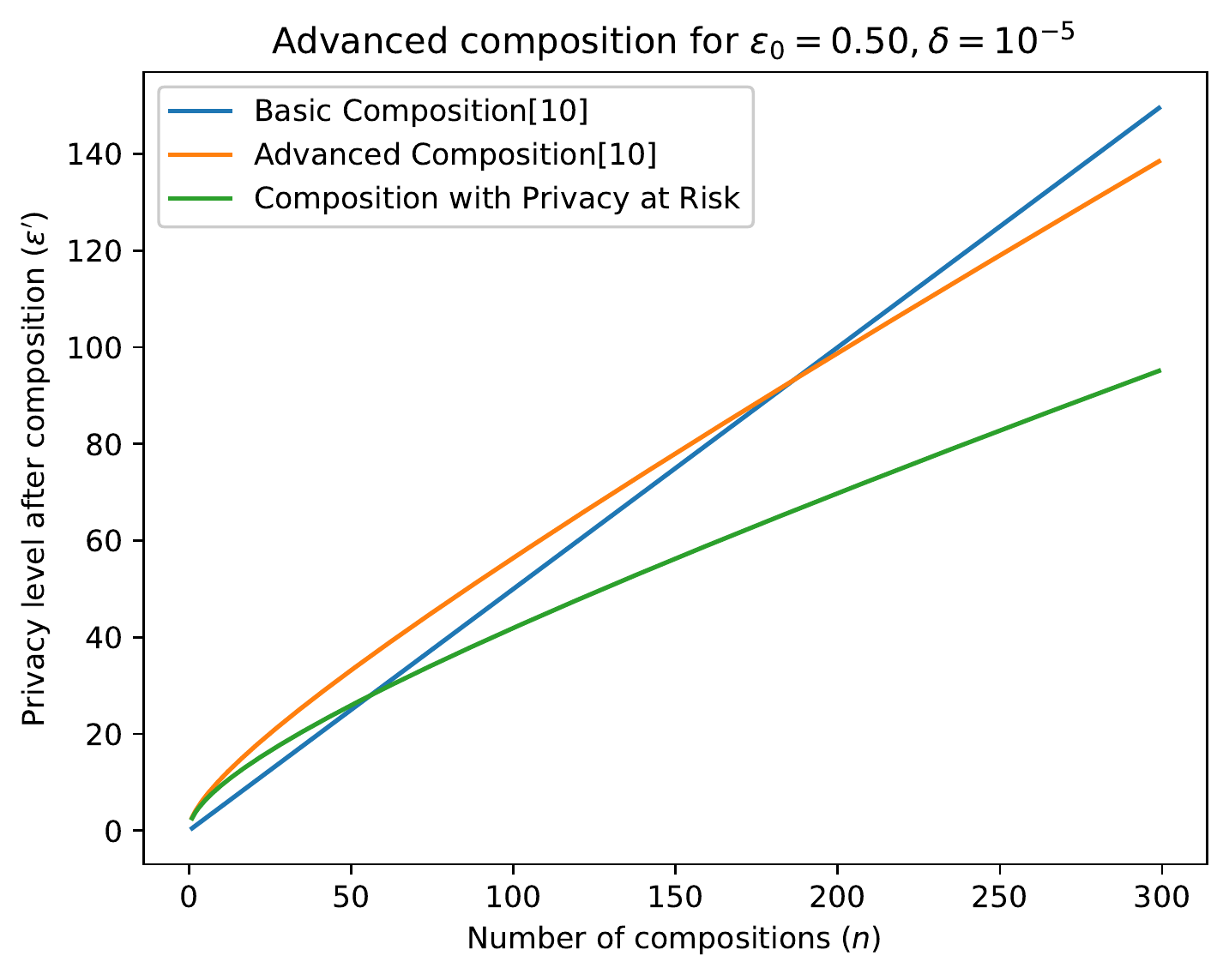}
  \captionof{figure}{$\mathcal{L}_{0.5}^1$ satisfies $(0.27,0.61)$-privacy at risk.}
  \label{fig:composition_1}
\end{subfigure}%
\begin{subfigure}{0.33\textwidth}
  \centering
  \includegraphics[width=\textwidth]{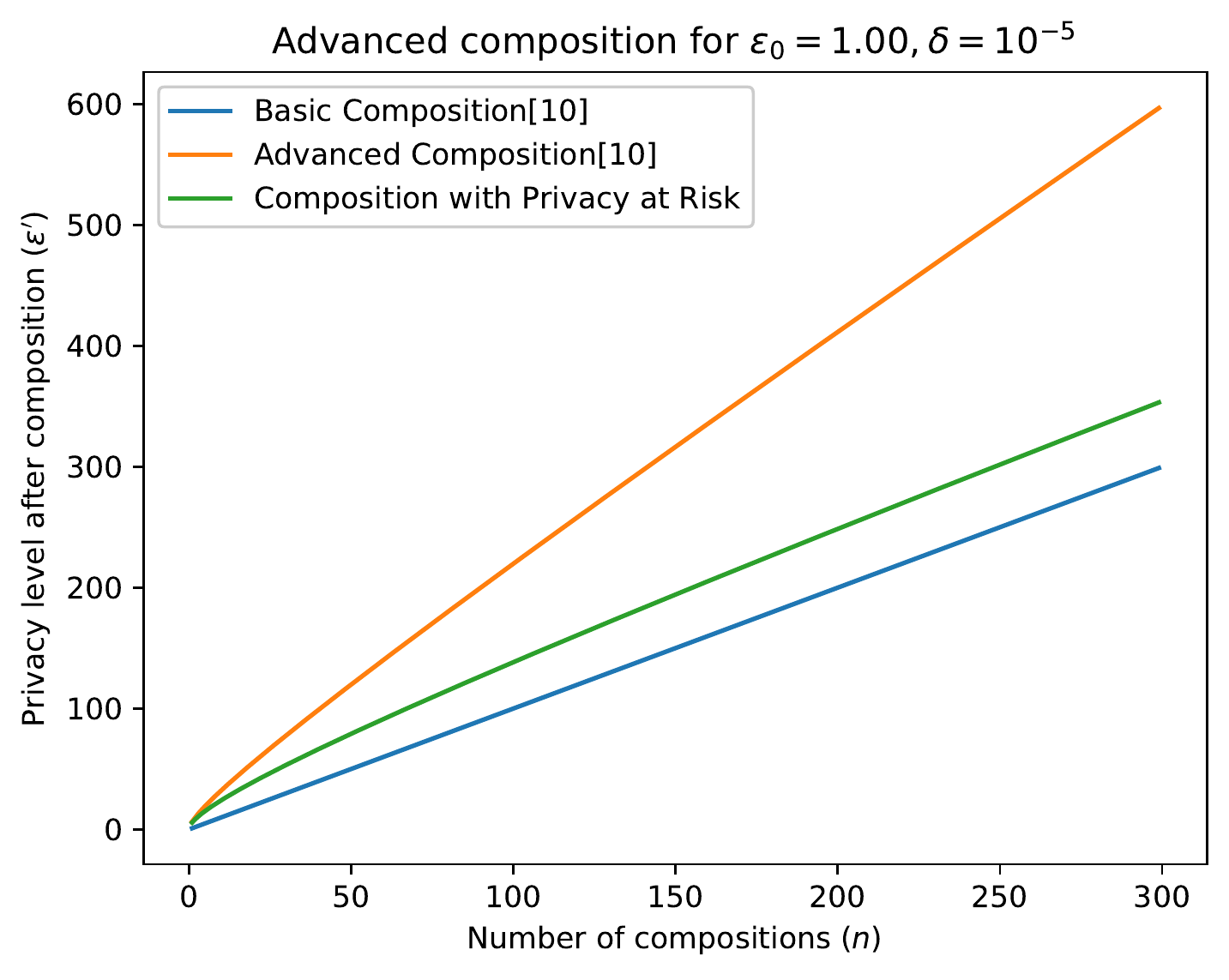}
  \captionof{figure}{$\mathcal{L}_{1.0}^1$ satisfies $(0.42,0.54)$-privacy at risk.}
  \label{fig:composition_2}
\end{subfigure}%
\caption{Comparing the privacy guarantee obtained by basic composition and advanced composition~\citep{dwork2014algorithmic} with the composition obtained using optimal privacy at risk that minimises the cost of Laplace mechanism $\mathcal{L}_{\epsilon_0}^1$. For the evaluation, we set $\delta = 10^{-5}$.}
\label{fig:composition}
\end{figure*}

\subsection{Cost Model and the Composition of Laplace Mechanisms}
Convexity of the proposed cost function enables us to estimate the optimal value of the privacy at risk level. We use the optimal privacy value to provide tighter bounds on the composition of Laplace mechanism. In Figure~\ref{fig:composition}, we compare the privacy guarantees obtained by using basic composition theorem~\citep{dwork2014algorithmic}, advanced composition theorem~\citep{dwork2014algorithmic} and the composition theorem for privacy at risk. We comparatively evaluate them for composition of Laplace mechanisms with privacy levels $0.1, 0.5$ and $1.0$. We compute the privacy level after composition by setting $\delta$ to $10^{-5}$. 

We observe that the use of optimal privacy at risk provided significantly stronger privacy guarantees as compared to the conventional composition theorems. Advanced composition theorem is known to provide stronger privacy guarantees for mechanism with smaller $\epsilon$s. As we observe in Figure~\ref{fig:composition_2} and Figure~\ref{fig:composition_1}, the composition provides strictly stronger privacy guarantees than basic composition, in the cases where the advanced composition fails.

\section{Balancing Utility and Privacy}
\label{sec:appl1}
In this section, we empirically illustrate and discuss the steps that a data steward needs to take and the issues that she needs to consider in order to realise a required privacy at risk level $\varepsilon$ for a confidence level $\gamma$ when seeking to disclose the result of a query.  

We consider a query that returns the parameter of a ridge regression~\citep{murphy} for an input dataset. It is a basic and widely used statistical analysis tool. We use the privacy-preserving mechanism presented by~\citeauthor{ligett2017accuracy} for ridge regression. It is a Laplace mechanism that induces noise in the output parameters of the ridge regression. The authors provide a theoretical upper bound on the sensitivity of the ridge regression, which we refer as \textit{sensitivity}, in the experiments. 

\subsection{Dataset and Experimental Setup}
We conduct experiments on a subset of the 2000 US census dataset provided by Minnesota Population Center in its Integrated Public Use Microdata Series~\citep{IPUMS}. The census dataset consists of 1\% sample of the original census data. It spans over 1.23 million households with records of 2.8 million people. The value of several attributes is not necessarily available for every household. We have therefore selected $212,605$ records, corresponding to the household heads, and $6$ attributes, namely, \textit{Age, Gender, Race, Marital Status, Education, Income}, whose values are available for the $212,605$ records. 

In order to satisfy the constraint in the derivation of the sensitivity of ridge regression~\citep{ligett2017accuracy}, we, without loss of generality, normalise the dataset in the following way. We normalise \textit{Income} attribute such that the values lie in $[0,1]$. We normalise other attributes such that $l_2$ norm of each data point is unity.

All experiments are run on Linux machine with 12-core 3.60GHz Intel\textsuperscript{\textregistered} Core i7\texttrademark processor with 64GB memory. Python\textsuperscript{\textregistered} 2.7.6 is used as the scripting language.

\subsection{Result Analysis}

We train ridge regression model to predict \textit{Income} using other attributes as predictors. We split the dataset into the training dataset ($80\%$) and testing dataset ($20\%$). We compute the \textit{root mean squared error (RMSE)} of ridge regression, trained on the training data with regularisation parameter set to $0.01$, on the testing dataset. We use it as the metric of \textit{utility loss}. Smaller the value of RMSE, smaller the loss in utility. For a given value of privacy at risk level, we compute $50$ runs of an experiment of a differentially private ridge regression and report the means over the $50$ runs of the experiment.

Let us now provide illustrative experiments under the three different cases. In every scenario, the data steward is given a privacy at risk level $\epsilon$ and the confidence level $\gamma$ and wants to disclose the parameters of a ridge regression model that she trains on the census dataset. She needs to calibrate the Laplace mechanism to achieve the privacy at risk required the ridge regression query.

\textbf{The Case of Explicit  Randomness (cf. Section~\ref{sec:case_one}).} In this scenario, the data steward knows the sensitivity for the ridge regression. She needs to compute the privacy level, $\epsilon_0$, to calibrate the Laplace mechanism. She uses Equation~\ref{eqn:gamma_1} that links the desired privacy at risk level $\epsilon$, the confidence level $\gamma_1$ and the privacy level of noise $\epsilon_0$. Specifically, for given $\epsilon$ and $\gamma_1$, she computes $\epsilon_0$ by solving the equation:
\[
\gamma_1 \mathbb{P}(T \leq \epsilon_0) - \mathbb{P}(T \leq \epsilon)  = 0.
\]
Since the equation does not give an analytical formula for $\epsilon_0$, the data steward uses a root finding algorithm such as Newton-Raphson method~\citep{press2007numerical} to solve the above equation. For instance, if she needs to achieve a privacy at risk level $\epsilon = 0.4$ with confidence level $\gamma_1 = 0.6$, she can substitute these values in the above equation and solve the equation to get the privacy level of noise $\epsilon_0 = 0.8$.

\begin{figure*}
	\centering
	\hspace*{-0.5em}
	\begin{minipage}{.5\linewidth}
		\centering
		\includegraphics[width=\columnwidth]{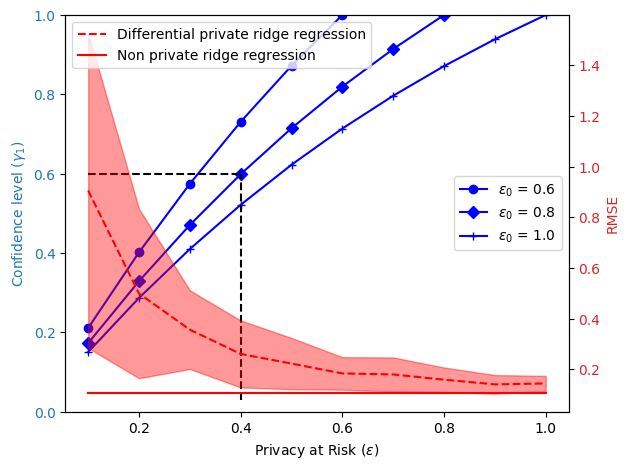}
		\captionof{figure}{Utility, measured by RMSE (right y-axis), and privacy at risk level $\epsilon$ for Laplace mechanism (left y-axis) for varying confidence levels $\gamma_1$.}
		\label{fig:case1}
	\end{minipage}%
	\hspace*{0.5em}
	\begin{minipage}{.5\linewidth}
		\centering
		\includegraphics[width=0.95\columnwidth]{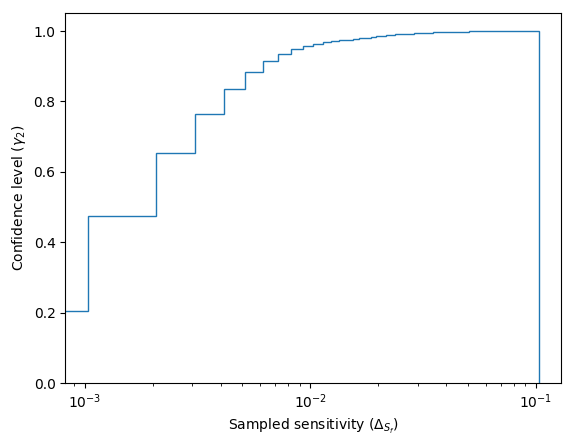}
		\captionof{figure}{Empirical cumulative distribution of the sensitivities of ridge regression queries constructed using $15000$ samples of neighboring datasets.}
		\label{fig:sensitivity_histogram}
	\end{minipage}
\end{figure*}

Figure~\ref{fig:case1} shows the variation of privacy at risk level $\epsilon$ and confidence level $\gamma_1$.  It also depicts the variation of utility loss for different privacy at risk levels in Figure~\ref{fig:case1}. 

In accordance to the data steward's problem, if she needs to achieve a privacy at risk level $\epsilon = 0.4$ with confidence level $\gamma_1 = 0.6$, she obtains the privacy level of noise to be $\epsilon_0 = 0.8$. Additionally, we observe that the choice of privacy level $0.8$ instead of $0.4$ to calibrate the Laplace mechanism gives lower utility loss for the data steward. This is the benefit drawn from the risk taken under the control of privacy at risk.

Thus, she uses privacy level $\epsilon_0$ and the sensitivity of the function to calibrate Laplace mechanism.

\textbf{The Case of Implicit Randomness (cf. Section~\ref{sec:case_two}).} In this scenario, the data steward does not know the sensitivity of ridge regression. She assesses that she can afford to sample at most $n$ times from the population dataset. She understands the effect of the uncertainty introduced by the statistical estimation of the sensitivity. Therefore, she uses the confidence level for empirical privacy at risk $\hat{\gamma_2}$. 

Given the value of $n$, she chooses the value of the accuracy parameter using  Figure~\ref{fig:sample_size}. For instance, if the number of samples that she can draw is $10^4$, she chooses the value of the accuracy parameter $\rho = 0.01$. Next, she uses Equation~\ref{eqn:prob_tolerance} to determine the value of probabilistic tolerance, $\alpha$, for the sample size $n$. For instance, if the data steward is not allowed to access more than $15,000$ samples, for the accuracy of $0.01$ the probabilistic tolerance is $0.9$.
\begin{equation}
\alpha = 1 - 2e^{(-2\rho^2n)}
\label{eqn:prob_tolerance}
\end{equation}
She constructs an empirical cumulative distribution over the sensitivities as described in Section~\ref{sec:case_two}. Such an empirical cumulative distribution is shown in Figure~\ref{fig:sensitivity_histogram}. Using the computed probabilistic tolerance and desired confidence level $\hat{\gamma_2}$, she uses equation in Theorem~\ref{thm:par_2} to determine $\gamma_2$. She computes the sampled sensitivity using the empirical distribution function and the confidence level for privacy $\Delta_{S_f}$ at risk $\gamma_2$. For instance, using the empirical cumulative distribution in Figure~\ref{fig:sensitivity_histogram} she calculates the value of the sampled sensitivity to be approximately $0.001$ for $\gamma_2 = 0.4$ and approximately $0.01$ for $\gamma_2 = 0.85$

Thus, she uses privacy level $\epsilon$, sets the number of samples to be $n$ and computes the sampled sensitivity $\Delta_{S_f}$ to calibrate the Laplace mechanism.

\textbf{The Case of Explicit and Implicit Randomness (cf. Section~\ref{sec:case_three}).} In this scenario, the data steward does not know the sensitivity of ridge regression. She is not allowed to sample more than $n$ times from a population dataset. For a given confidence level $\gamma_2$ and the privacy at risk $\epsilon$, she calibrates the Laplace mechanism using illustration for Section~\ref{sec:case_three}. The privacy level in this calibration yields utility loss that is more than her requirement. Therefore, she wants to re-calibrate the Laplace mechanism in order to reduce utility loss. 

For the re-calibration, the data steward uses privacy level of the pre-calibrated Laplace mechanism, i.e. $\epsilon$, as the privacy at risk level and she provides a new confidence level for empirical privacy at risk $\hat{\gamma_3}$. Using Equation~\ref{eqn:gamma_3} and Equation~\ref{eqn:13}, she calculates:
\[
\hat{\gamma_3} \mathbb{P}(T \leq \eta\epsilon_0) - \alpha\gamma_2~ \mathbb{P}(T \leq \epsilon) = 0
\]
She solves such an equation for $\epsilon_0$ using the root finding technique such as Newton-Raphson method~\citep{press2007numerical}.
For instance, if she needs to achieve a privacy at risk level $\epsilon = 0.4$ with confidence levels $\hat{\gamma_3} = 0.9$ and $\gamma_2 = 0.9$, she can substitute these values and the values of tolerance parameter and sampled sensitivity, as used in the previous experiments, in the above equation. Then, solving the equation  leads to the privacy level of noise $\epsilon_0 = 0.8$.

Thus, she re-calibrates the Laplace mechanism with privacy level $\epsilon_0$, sets the number of samples to be $n$ and sampled sensitivity $\Delta_{S_f}$.

\section{Related Work}
\label{sec:related}

\textbf{Calibration of mechanisms.} Researchers have proposed different privacy-preserving mechanisms to make different queries differentially private. These mechanisms can be broadly classified into two categories. In one category, the mechanisms  explicitly add calibrated noise, such as Laplace noise in the work of~\citep{laplace_mechanism} or Gaussian noise in the work of~\citep{dwork2014algorithmic}, to the outputs of the query. In the other category, \citep{chaudhuri2011differentially,functional,fourier,RKHS} propose mechanisms that alter the query function so that the modified function satisfies differentially privacy. 
Privacy-preserving mechanisms in both of these categories perturb the original output of the query and make it difficult for a malicious data analyst to recover the original output of the query. These mechanisms induce randomness using the explicit noise distribution. Calibration of these mechanisms require the knowledge of the sensitivity of the query. \citeauthor{nissim2007smooth} consider the implicit randomness in the data-generation distribution to compute an estimate of the sensitivity. The authors propose the smooth sensitivity function that is an envelope over the local sensitivities for all individual datasets. Local sensitivity of a dataset is the maximum change in the value of the query over all of its neighboring datasets. In general, it is not easy to analytically estimate the smooth sensitivity function for a general query. \citeauthor{rubinstein2017pain} also study the inherent randomness in the data-generation algorithm. They do not use the local sensitivity. We adopt their approach of sampling the sensitivity from the empirical distribution of the sensitivity. They use order statistics to choose a particular value of the sensitivity. We use the risk, which provides a mediation tool for business entities to assess the actual business risks, on the sensitivity distribution to estimate the sensitivity.

\textbf{Refinements of differential privacy.} In order to account for both sources of randomness, refinements of $\epsilon$-differential privacy are proposed in order to bound the probability of occurrence of worst case scenarios. 
\citeauthor{machanavajjhala2008privacy} propose probabilistic differential privacy that considers upper bounds of the worst case privacy loss for corresponding confidence levels on the noise distribution. Definition of  probabilistic differential privacy incorporates the explicit randomness induced by the noise distribution and bounds the probability over the space of noisy outputs to satisfy the $\epsilon$-differential privacy definition.
\citeauthor{dwork2016concentrated} propose Concentrated differential privacy that considers the expected values of the privacy loss random variables for the corresponding. Definition of concentrated differential privacy incorporates the explicit randomness induced by the noise distribution but considering only the expected value of privacy loss satisfying $\epsilon$-differential privacy definition instead of using the confidence levels limits its scope.
	
\citeauthor{RKHS} propose random differential privacy that considers the privacy loss for corresponding confidence levels on the implicit randomness in the data-generation distribution. Definition of random differential privacy incorporates the implicit randomness induced by the data-generation distribution and bounds the probability over the space of datasets generated from the given distribution to satisfy the $\epsilon$-differential privacy definition.
\citeauthor{dwork2006our} define approximate differential privacy by adding a constant bias to the privacy guarantee provided by the differential privacy. It is not a probabilistic refinement of the differential privacy.

Around the same time of our work, \citeauthor{triastcyn2019federated} independently propose Bayesian differential privacy that takes into account both of the sources of randomness. Despite this similarity, our works differ in multiple dimensions. Firstly, they have shown the reduction of their definition to a variant of Renyi differential privacy that depends on the data-generation distribution. Secondly, they rely on the moment accountant for the composition of the mechanisms. Lastly, they do not provide a finer case-by-case analysis of the source of randomness, which leads to analytical solutions for the privacy guarantee. 

\citeauthor{kifer2012rigorous} define Pufferfish privacy framework, and its variant by \citeauthor{bassily2013coupled}, that considers randomness due to data-generation distribution as well as noise distribution. Despite the generality of their approach, the framework relies on the domain expert to define a set of \emph{secrets} that they want to protect.

\textbf{Composition theorem.} Recently proposed technique of the \textit{moment accountant}~\citep{abadi2016deep} has become the state-of-the-art of composing mechanisms in the area of privacy-preserving machine learning. Abadi et al. show that the moment accountant provides much strong privacy guarantees than the conventional composition mechanisms. It works by keeping track of various moments of privacy loss random variable and use the bounds on them to provide privacy guarantees. The moment accountant requires access to data-generation distribution to compute the bounds on the moment. Hence, the privacy guarantees are specific to the dataset.

% In this work, we consider the widely used Laplace mechanism proposed by~\citep{laplace_mechanism}. The Laplace mechanism adds Laplacian noise to the query output. \citep{xiao2011differential,functional,fourier} use Laplace mechanism by providing the calibration by computing sensitivity of the query.
 
\textbf{Cost models.} \citep{ghosh2015selling,chen2016truthful} propose game theoretic methods that provide the means to evaluate the monetary cost of differential privacy. Our approach is inspired by the approach in the work of~\citeauthor{hsu2014differential}. They model the cost under a scenario of a research study wherein the participants are reimbursed for their participation. Our cost modelling is driven by the scenario of securing a compensation budget in compliance with GDPR. Our requirement differs from the requirements for the scenario in their work. In our case, there is no monetary incentive for participants to share their data.

\section{Conclusion and Future Works}
In this paper, we provide a means to fine-tune the privacy level of a privacy-preserving mechanism by analysing various sources of randomness. Such a fine-tuning leads to probabilistic quantification on privacy levels with quantified risks, which we call as privacy at risk. We also provide composition theorem that leverages privacy at risk. We analytical calculate privacy at risk for Laplace mechanism. We propose a cost model that bridges the gap between the privacy level and the compensation budget estimated by a GDPR compliant business entity. Convexity of the cost function ensures existence of unique privacy at risk that minimises compensation budget. The cost model helps in not only reinforcing the ease of application in a business setting but also providing stronger privacy guarantees on the composition of mechanism.

Privacy at risk may be fully analytically computed in cases where the data-generation, or the sensitivity distribution, the noise distribution and the query are analytically known and take convenient forms. We are now looking at such convenient but realistic cases.

\section*{Acknowledgements}
We want convey a special thanks to Pierre Senellart at DI, {\'E}cole Normale Sup{\'e}rieure, Paris for his careful reading of our drafts and thoughtful interventions.

\bibliography{ref} 

\newpage
\appendix
\section{Proof of Theorem~\ref{thm:par_1} (Section~\ref{sec:case_one})}
\label{app:case_1}
Although a Laplace mechanism $\mathcal{L}^{\Delta_f}_{\epsilon}$ induces higher amount of noise on average than a Laplace mechanism $\mathcal{L}^{\Delta_f}_{\epsilon_0}$ for $\epsilon < \epsilon_0$, there is a non-zero probability that $\mathcal{L}^{\Delta_f}_{\epsilon}$ induces noise commensurate to $\mathcal{L}^{\Delta_f}_{\epsilon_0}$. This non-zero probability guides us to calculate the privacy at risk $\gamma_1$ for the privacy at risk level $\epsilon$. In order to get an intuition, we illustrate the calculation of the overlap between two Laplace distributions as an estimator of similarity between the two distributions.
\begin{definition}\textit[{Overlap of Distributions, \citep{papoulis2002probability}}]
The overlap, $O$, between two probability distributions $P_1, P_2$ with support $\mathcal{X}$ is defined as
\[
O = \int_{\mathcal{X}} \min_{} [P_1(x), P_2(x)] ~{dx}.
\]
\end{definition}
\begin{lemma}
The overlap $O$ between two probability distributions, $\mathrm{Lap}(\frac{\Delta_f}{\epsilon_1})$ and $\mathrm{Lap}(\frac{\Delta_f}{\epsilon_2})$, such that $\epsilon_2 \leq \epsilon_1$, is given by
\[
O = 1 - (\exp{(-\mu \epsilon_2 / \Delta_f)} - \exp{(-\mu \epsilon_1 / \Delta_f)}),
\]
where $\mu = \frac{\Delta_f \ln{(\epsilon_1 / \epsilon_2)}}{\epsilon_1 - \epsilon_2}$.
\label{lemma:overlap}
\end{lemma}

Using the result in Lemma~\ref{lemma:overlap}, we note that the overlap between two distributions with $\epsilon_0 = 1$ and $\epsilon = 0.6$ is $0.81$. Thus, $\mathcal{L}^{\Delta_f}_{0.6}$ induces noise that is more than $80\%$ times similar to the noise induced by $\mathcal{L}^{\Delta_f}_{1.0}$.  Therefore, we can loosely say that at least $80\%$ of the times a Laplace Mechanism $\mathcal{L}^{\Delta_f}_{1.0}$ will provide the same privacy as a Laplace Mechanism $\mathcal{L}^{\Delta_f}_{0.8}$.

Although the overlap between Laplace distributions with different scales offers an insight into the relationship between different privacy levels, it does not capture the constraint induced by the \textit{sensitivity}. For a given query $f$, the amount of noise required to satisfy differential privacy is commensurate to the sensitivity of the query. This calibration puts a constraint on the noise that is required to be induced on a pair of neighbouring datasets. We state this constraint in Lemma~\ref{lemma:bound}, which we further use to prove that the Laplace Mechanism $\mathcal{L}^{\Delta_f}_{\epsilon_0}$ satisfies $(\epsilon, \gamma_1)$-privacy at risk. 

\begin{lemma}
For a Laplace Mechanism $\mathcal{L}^{\Delta_f}_{\epsilon_0}$, the difference in the absolute values of noise induced on a pair of neighbouring datasets is upper bounded by the sensitivity of the query.
\label{lemma:bound}
\end{lemma}
\begin{proof}
Suppose that two neighbouring datasets $x$ and $y$ are given input to a numeric query $f: \mathcal{D} \rightarrow \mathbb{R}^k$. For any output $z \in \mathbb{R}^k$ of the Laplace Mechanism $\mathcal{L}^{\Delta_f}_{\epsilon_0}$,
\begin{align*}
\sum_{i=1}^k\left(|f(y_i) - z_i| - |f(x_i) - z_i|\right) &\leq \sum_{i=1}^k \left(|f(x_i) - f(y_i)|\right) \\
&\leq \Delta_f.
\end{align*}
We use triangular inequality in the first step and Definition~\ref{def:sensitivity} of sensitivity in the second step.
\end{proof}

We write $\mathrm{Exp}(b)$ to denote a random variable sampled from an \emph{exponential distribution} with scale $b > 0$. We write $\mathrm{Gamma}(k, \theta)$ to denote a random variable sampled from a \emph{gamma distribution} with shape $k > 0$ and scale $\theta > 0$.
\begin{lemma}\textit[{\citep{papoulis2002probability}}]
If a random variable $X$ follows Laplace Distribution with mean zero and scale $b$, $|X| \sim \mathrm{Exp}(b)$.
\label{lemma:exp}
\end{lemma}

\begin{lemma}\textit[{\citep{papoulis2002probability}}]
If $X_1, ..., X_n$ are $n$ i.i.d. random variables each following the Exponential Distribution with scale $b$, $\sum_{i=1}^n X_i \sim \textrm{Gamma}(n, b)$.
\label{lemma:gamma}
\end{lemma}

\begin{lemma}
If $X_1$ and $X_2$ are two i.i.d. Gamma$(n, \theta)$ random variables, the probability density function for the random variable $T = |X_1 - X_2| / \theta$ is given by
\[
P_{T}(t; n, \theta) = \frac{2^{2-n}t^{n- \frac{1}{2}}K_{n- \frac{1}{2}}(t)}{\sqrt{2\pi} \Gamma(n) \theta}
\]
where $K_{n- \frac{1}{2}}$ is the modified Bessel function of second kind.
\vspace*{-0.5em}
\label{lemma:besselK}
\end{lemma}
\begin{proof}
Let $X_1$ and $X_2$ be two i.i.d. $\mathrm{Gamma}(n, \theta)$ random variables. Characteristic function of a Gamma random variable is given as
\[
\phi_{X_1}(z) = \phi_{X_2}(z) = (1 - \iota z\theta)^{-n}.
\]
Therefore,
\[
\phi_{X_1 - X_2}(z) = \phi_{X_1}(z) \phi_{X_2}^*(z) = \frac{1}{(1 + (z\theta)^2)^n}
\]
Probability density function for the random variable $X_1 - X_2$ is given by,
\begin{align*}
P_{X_1 - X_2}(x) &= \frac{1}{2\pi} \int_{-\infty}^\infty e^{-\mathrm{i}zx}\phi_{X_1 - X_2}(z) dz \\
&= \frac{2^{1-n}{|\frac{x}{\theta}|}^{n- \frac{1}{2}}K_{n- \frac{1}{2}}(|\frac{x}{\theta}|)}{\sqrt{2\pi} \Gamma(n) \theta}
\end{align*}
where $K_{n- \frac{1}{2}}$ is the Bessel function of second kind.
Let $T = |\frac{X_1 - X_2}{\theta}|$. Therefore,
\[
P_{T}(t; n, \theta) = \frac{2^{1-n}t^{n- \frac{1}{2}}K_{n- \frac{1}{2}}(t)}{\sqrt{2\pi} \Gamma(n) \theta}
\]
We use Mathematica~\citep{Mathematica} to solve the above integral.
% We denote this probability distribution as $\mathrm{BesselK}(n, \theta)$.
\end{proof}

\begin{lemma}
If $X_1$ and $X_2$ are two i.i.d. Gamma$(n, \theta)$ random variables and $|X_1 - X_2| \leq M$, then $T' = |X_1 - X_2| / \theta$ follows 
% Truncated BesselK$(n, \theta, M)$ 
the distribution with probability density function:
\[
P_{T'}(t; n, \theta, M) = \frac{P_T(t'; n, \theta)}{P_T(T \leq M)},
\]
where $P_T$ is the probability density function of defined in Lemma~\ref{lemma:besselK}.
\label{lemma:trunc_bessel}
\end{lemma}

\begin{lemma}
For Laplace Mechanism $\mathcal{L}^{\Delta_f}_{\epsilon_0}$ with query $f: \mathcal{D} \rightarrow \mathbb{R}^k$ and for any output $Z \subseteq Range(\mathcal{L}^{\Delta_{f}}_{\epsilon_0})$, $\epsilon \leq \epsilon_0$,
\[
\gamma_1 \triangleq \mathbb{P}\left[\ln{\left |\frac{\mathbb{P}(\mathcal{L}^{\Delta_f}_{\epsilon_0}(x) \in Z)}{\mathbb{P}(\mathcal{L}^{\Delta_f}_{\epsilon_0}(y) \in Z)}\right|} \leq \epsilon \right] = \frac{\mathbb{P}(T \leq \epsilon)}{\mathbb{P}(T \leq \epsilon_0)},
\]
where $T$ follows the distribution in Lemma~\ref{lemma:besselK}, $P_T(t; k, \frac{\Delta_f}{\epsilon_0})$.
\label{lemma:main}
\end{lemma}
\begin{proof}
Let, $x \in \mathcal{D}$ and $y \in \mathcal{D}$ be two datasets such that $x \sim y$. Let $f: \mathcal{D} \rightarrow \mathbb{R}^k$ be some numeric query. Let $\mathbb{P}_x(z)$ and $\mathbb{P}_y(z)$ denote the probabilities of getting the output $z$ for Laplace mechanisms $\mathcal{L}^{\Delta_f}_{\epsilon_0}(x)$ and $\mathcal{L}^{\Delta_f}_{\epsilon_0}(y)$ respectively. For any point $z \in \mathbb{R}^k$ and $\epsilon \neq 0$,
\begin{align}~\label{eqn:3}
\frac{\mathbb{P}_x(z)}{\mathbb{P}_y(z)}  &= \prod_{i=1}^k \frac {\exp{\left( \frac{-\epsilon_0 |f(x_i) - z_i|}{\Delta_f} \right)}} {\exp{\left( \frac{-\epsilon_0 |f(y_i) - z_i|}{\Delta_f} \right)}} \nonumber \\
&= \prod_{i=1}^k \exp{\left( \frac{\epsilon_0 (|f(y_i) - z_i| - |f(x_i) - z_i|)}{\Delta_f} \right)} \nonumber \\
&= \exp{\left(\epsilon \left[ \frac{\epsilon_0 \sum_{i=1}^k(|f(y_i) - z_i| - |f(x_i) - z_i|)}{\epsilon \Delta_f} \right] \right)} .
\end{align}
By Definition~\ref{def:lap_mech},
\begin{equation}
\label{eqn:4}
(f(x) - z), (f(y) - z) \sim \textrm{Lap}(\Delta_f / \epsilon_0).
\end{equation}
Application of Lemma~\ref{lemma:exp} and Lemma~\ref{lemma:gamma} yields,
\begin{equation}
\sum_{i=1}^k \left(|f(x_i) - z_i|\right) \sim \textrm{Gamma}(k, \Delta_f / \epsilon_0).
\label{eqn:5}
\end{equation}
Using Equations~\ref{eqn:4},~\ref{eqn:5}, and Lemma~\ref{lemma:bound}, ~\ref{lemma:trunc_bessel}, we get
\begin{multline}
    \left(\frac{\epsilon_0}{\Delta_f} \sum_{i=1}^k \left | \left(|f(y_i) - z| - |f(x_i) - z|\right)\right |  \right) \\
    \sim P_{T'}(t; k, \Delta_f / \epsilon_0, \Delta_f).
\label{eqn:for_corollary2}
\end{multline}
% \begin{equation}
% \left(\frac{\epsilon_0}{\Delta_f} \sum_{i=1}^k \left | \left(|f(y_i) - z| - |f(x_i) - z|\right)\right |  \right) \sim \mathrm{Truncated BesselK}(k, \Delta_f / \epsilon_0, \Delta_f).
% \label{eqn:for_corollary2}
% \end{equation}
since, $\sum_{i=1}^k \left | \left(|f(y_i) - z| - |f(x_i) - z|\right)\right | \leq \Delta_f$.
Therefore,
\begin{equation}
\mathbb{P}\left(\left[\frac{\epsilon_0}{\Delta_f} \sum_{i=1}^k\ \left | \left(|f(y_i) - z| - |f(x_i) - z|\right)\right| \right] \leq \epsilon \right) = \frac{\mathbb{P}(T \leq \epsilon)}{\mathbb{P}(T \leq \epsilon_0)},
\label{eqn:6}
\end{equation}
where $T$ follows the distribution in Lemma~\ref{lemma:besselK}. We use Mathematica~\citep{Mathematica} to analytically compute,
\begin{multline*}
     \mathbb{P}(T \leq x) \propto \left(_1F_2(\frac{1}{2}; \frac{3}{2}-k,\frac{3}{2};\frac{x^2}{4}) \sqrt{\pi} 4^k x]\right) - \\ \left(2 _1F_2(k; \frac{1}{2}+k,k+1;\frac{x^2}{4})x^{2k} \Gamma(k)\right)
\end{multline*}
% \[
% \mathbb{P}(T \leq x) \propto \left[_1F_2(\frac{1}{2}; \frac{3}{2}-k,\frac{3}{2};\frac{x^2}{4}) \sqrt{\pi} 4^k x - 2 _1F_2(k; \frac{1}{2}+k,k+1;\frac{x^2}{4})x^{2k} \Gamma(k)\right],
% \]
where $_1F_2$ is the regularised generalised hypergeometric function as defined in~\citep{askey2010generalized}.
From Equation~\ref{eqn:3} and~\ref{eqn:6},
\[
\mathbb{P}\left[\ln{\left |\frac{\mathbb{P}(\mathcal{L}_{\epsilon_0}^{\Delta_f}(x) \in S)}{\mathbb{P}(\mathcal{L}_{\epsilon_0}^{\Delta_f}(y) \in S)}\right|} \leq \epsilon \right] = \frac{\mathbb{P}(T \leq \epsilon)}{\mathbb{P}(T \leq \epsilon_0)}.
\]
\end{proof}

This completes the proof of Theorem~\ref{thm:par_1}.

\begin{corollary}
Laplace Mechanism $\mathcal{L}^{\Delta_f}_{\epsilon_0}$ with $f: \mathcal{D} \rightarrow \mathbb{R}^k$ is $(\epsilon, \delta)$-probabilistically differentially private where
\[
\delta = 
	\begin{cases}
		1 - \frac{\mathbb{P}(T \leq \epsilon)}{\mathbb{P}(T \leq \epsilon_0)} &\quad \epsilon\leq\epsilon_0 \\
        0 &\quad \epsilon > \epsilon_0
	\end{cases}
\]
and $T$ follows $\mathrm{BesselK}(k, \Delta_f / \epsilon_0)$.
\label{lemma:ed_privacy}
\end{corollary}

\section{Proof of Theorem~\ref{thm:par_2} (Section~\ref{sec:case_two})}
\label{app:case_2}

\begin{proof}
Let, $x$ and $y$ be any two neighbouring datasets sampled from the data generating distribution $\mathcal{G}$. Let, $\Delta_{S_f}$ be the sampled sensitivity for query $f: \mathcal{D} \rightarrow \mathbb{R}^k$. Let, $\mathbb{P}_x(z)$ and $\mathbb{P}_y(z)$ denote the probabilities of getting the output $z$ for Laplace mechanisms $\mathcal{L}_{\epsilon}^{\Delta_{S_f}}(x)$ and $\mathcal{L}_{\epsilon}^{\Delta_{S_f}}(y)$ respectively. For any point $z \in \mathbb{R}^k$ and $\epsilon \neq 0$,
\begin{align}
\frac{\mathbb{P}_x(z)}{\mathbb{P}_y(z)} &= \prod_{i=1}^k \frac {\exp{\left( \frac{-\epsilon |f(x_i) - z_i|}{\Delta_{S_f}} \right)}} {\exp{\left( \frac{-\epsilon |f(y_i) - z_i|}{\Delta_{S_f}} \right)}} \nonumber \\
&= \exp{\left( \frac{\epsilon \sum_{i=1}^k(|f(y_i) - z_i| - |f(x_i) - z_i|)}{\Delta_{S_f}} \right)} \nonumber \\
&\leq \exp{\left( \frac{\epsilon \sum_{i=1}^k|f(y_i) - f(x_i)|}{\Delta_{S_f}} \right)} \nonumber \\ 
&= \exp{\left( \frac{\epsilon \lVert f(y) - f(x) \rVert_1}{\Delta_{S_f}} \right)} \label{eqn:7}
\end{align}
We used triangle inequality in the penultimate step. 

Using the trick in the work of~\citep{rubinstein2017pain}, we define following events. Let, $B^{\Delta_{S_f}}$ denotes the set of pairs neighbouring dataset sampled from $\mathcal{G}$ for which the sensitivity random variable is upper bounded by $\Delta_{S_f}$. Let, $C_\rho^{\Delta_{S_f}}$ denotes the set of sensitivity random variable values for which $F_n$ deviates from the unknown cumulative distribution of $S$, $F$, at most by the accuracy value $\rho$. These events are defined in Equation~\ref{eqn:events}.
\begin{align}
B^{\Delta_{S_f}} &\triangleq \{x, y \sim \mathcal{G} ~\textrm{such that}~ \lVert f(y) - f(x) \rVert_1 \leq \Delta_{S_f}\} \nonumber \\
C_\rho^{\Delta_{S_f}} &\triangleq \left\lbrace \sup_{\Delta} |F_S^n(\Delta) - F_S(\Delta)| \leq \rho \right\rbrace \label{eqn:events}
\end{align}

\begin{align}
    \mathbb{P}(B^{\Delta_{S_f}}) &= \mathbb{P}(B^{\Delta_{S_f}} | C_\rho^{\Delta_{S_f}}) \mathbb{P}(C_\rho^{\Delta_{S_f}})\\
    &+ \mathbb{P}(B^{\Delta_{S_f}} | \overline{C\rho^{\Delta_{S_f}}}) \mathbb{P}(\overline{C_\rho^{\Delta_{S_f}}}) \nonumber \\
&\geq \mathbb{P}(B^{\Delta_{S_f}} | C_\rho^{\Delta_{S_f}}) \mathbb{P}(C_\rho^{\Delta_{S_f}}) \nonumber \\
&= F_n(\Delta_{S_f}) \mathbb{P}(C_\rho^{\Delta_{S_f}}) \nonumber \\
&\geq \gamma_2 \cdot (1 - 2e^{-2\rho^2n}) \label{eqn:8}
\end{align}
In the last step, we use the definition of the sampled sensitivity to get the value of the first term. The last term is obtained using DKW-inequality, as defined in ~\citep{massart1990tight}, where the $n$ denotes the number of samples used to build empirical distribution of the sensitivity, $F_n$.

From Equation~\ref{eqn:7}, we understand that if $\lVert f(y) - f(x) \rVert_1$ is less than or equals to the sampled sensitivity then the Laplace mechanism $\mathcal{L}_\epsilon^{\Delta_{S_f}}$ satisfies $\epsilon$-differential privacy. Equation~\ref{eqn:8} provides the lower bound on the probability of the event $\lVert f(y) - f(x) \rVert_1 \leq \Delta_{S_f}$. Thus, combining Equation~\ref{eqn:7} and Equation~\ref{eqn:8} completes the proof.
\end{proof}

\section{Proof of Theorem~\ref{thm:par_3} (Section~\ref{sec:case_three})}
\label{app:case_3}
Proof of Theorem~\ref{thm:par_3} builds upon the ideas from the proofs for the rest of the two cases. In addition to the events defined in Equation~\ref{eqn:events}, we define an additional event $A_{\epsilon_0}^{\Delta_{S_f}}$, defined in Equation~\ref{eqn:event_a}, as a set of outputs of Laplace mechanism $\mathcal{L}_{\epsilon_0}^{\Delta_{S_f}}$ that satisfy the constraint of $\epsilon$-differential privacy for a specified privacy at risk level $\epsilon$.

\begin{equation}
A_{\epsilon_0}^{\Delta_{S_f}} \triangleq \left\lbrace z \sim \mathcal{L}_{\epsilon_0}^{\Delta_{S_f}} ~:~ \ln{\left| \frac{\mathcal{L}_{\epsilon_0}^{\Delta_{S_f}}(x)}{\mathcal{L}_{\epsilon_0}^{\Delta_{S_f}}(y)} \right| } \leq \epsilon, x, y \sim  \mathcal{G}\right\rbrace
\label{eqn:event_a}
\end{equation}

\begin{corollary}
\[
\mathbb{P}(A_{\epsilon_0}^{\Delta_{S_f}} |  B^{\Delta_{S_f}}) = \frac{\mathbb{P}(T \leq \epsilon)}{\mathbb{P}(T \leq \eta \epsilon_0)}
\]
where $T$ follows the distribution $P_T(t; \Delta_{S_f} / \epsilon_0)$ in Lemma~\ref{lemma:besselK} and $\eta = \frac{\Delta_f}{\Delta_{S_f}}$.
\label{cor:joint}
\end{corollary}
\begin{proof}
We provide the sketch of the proof. Proof follows from the proof of Lemma~\ref{lemma:main}. For a Laplace mechanism calibrated with the sampled sensitivity $\Delta_{S_f}$ and privacy level $\epsilon_0$, Equation~\ref{eqn:for_corollary2} translates to,
\begin{multline*}
    \left(\frac{\epsilon_0}{\Delta_{S_f}} \sum_{i=1}^k \left | \left(|f(y_i) - z| - |f(x_i) - z|\right)\right |  \right)
\sim \\ P_{T'}(t; k, \Delta_{S_f}/\epsilon_0, \Delta_{S_f}).
\end{multline*}

since, $\sum_{i=1}^k \left | \left(|f(y_i) - z| - |f(x_i) - z|\right)\right | \leq \Delta_f$.
Using Lemma~\ref{lemma:trunc_bessel} and Equation~\ref{eqn:6},
\[
\mathbb{P}(A_{\epsilon_0}^{\Delta_{S_f}}) = \frac{\mathbb{P}(T \leq \epsilon)}{\mathbb{P}(T \leq \eta \epsilon_0)}
\]
where $T$ follows the distribution $P_T(t; \Delta_{S_f} / \epsilon_0)$ and $\eta = \frac{\Delta_f}{\Delta_{S_f}}$.
\end{proof}

For this case, we do not assume the knowledge of the sensitivity of the query. Using the empirical estimation presented in Section~\ref{sec:case_two}, if we choose the sampled sensitivity for privacy at risk $\gamma_2 =1$, we obtain an approximation for $\eta$.
\begin{lemma}
For a given value of accuracy parameter $\rho$,
\[
\frac{\Delta_f}{\Delta_{S_f}^*} = 1 + \mathcal{O}\left(\frac{\rho}{\Delta_{S_f}^*}\right)
\]
where $\Delta_{S_f}^* = F_n^{-1}(1)$. $\mathcal{O}\left(\frac{\rho}{\Delta_{S_f}^*}\right)$ denotes order of $\frac{\rho}{\Delta_{S_f}^*}$, i.e., $\mathcal{O}\left(\frac{\rho}{\Delta_{S_f}^*}\right) = k \frac{\rho}{\Delta_{S_f}^*}$ for some $k \geq 1$.
\end{lemma}
\begin{proof}
For a given value of accuracy parameter $\rho$ and any $\Delta > 0$,
\[
F_n(\Delta) - F(\Delta) \leq \rho
\]
Since above inequality is true for any value of $\Delta$, let $\Delta = F^{-1}(1)$. Therefore,
\begin{align}
F_n(F^{-1}(1)) - F(F^{-1}(1)) &\leq \rho \nonumber \\
F_n(F^{-1}(1)) &\leq 1 + \rho \label{eqn:13}
\end{align}
Since a cumulative distribution function is $1$-Lipschitz [\citep{papoulis2002probability}], 
\begin{align*}
| F_n(F_n^{-1}(1)) - F_n(F^{-1}(1)) | &\leq |F_n^{-1}(1) - F^{-1}(1)|  \\
| F_n(F_n^{-1}(1)) - F_n(F^{-1}(1)) | &\leq |\Delta_{S_f}^* - \Delta_f|  \\
\rho &\leq \Delta_f - \Delta_{S_f}^* \\
1 + \frac{\rho}{\Delta_{S_f}^*} &\leq \frac{\Delta_f}{\Delta_{S_f}^*}
\end{align*}
where we used result from Equation~\ref{eqn:13} in step 3. Introducing $\mathcal{O}\left(\frac{\rho}{\Delta_{S_f}^*}\right)$ completes the proof.
\end{proof}
\vspace*{-1em}
\begin{lemma}
For Laplace Mechanism $\mathcal{L}_{\epsilon_0}^{\Delta_{S_f}}$ with sampled sensitivity $\Delta_{S_f}$ of a query $f: \mathcal{D} \rightarrow \mathbb{R}^k$ and for any $Z \subseteq Range(\mathcal{L}_\epsilon^{\Delta_{S_f}})$,
\[
\mathbb{P}\left[\ln{\left |\frac{\mathbb{P}(\mathcal{L}_{\epsilon_0}(x) \in Z)}{\mathbb{P}(\mathcal{L}_{\epsilon_0}(y) \in Z)}\right|} \leq \epsilon \right] \geq \frac{\mathbb{P}(T \leq \epsilon)}{\mathbb{P}(T \leq \eta \epsilon_0)} \gamma_2 (1 - 2e^{-2\rho^2n}) 
\]
where $n$ is the number of samples used to find sampled sensitivity, $\rho \in [0, 1]$ is a accuracy parameter and $\eta = \frac{\Delta_f}{\Delta_{S_f}}$. The outer probability is calculated with respect to support of the data-generation distribution $\mathcal{G}$.
\label{lemma:main_2}
\end{lemma}
\begin{proof}
The proof follows from the proof of Lemma~\ref{lemma:main} and Lemma~\ref{lemma:main_2}.
Consider,
\begin{align}
\mathbb{P}(A_{\epsilon_0}^{\Delta_{S_f}}) &\geq \mathbb{P}(A_{\epsilon_0}^{\Delta_{S_f}} | B^{\Delta_{S_f}}) \mathbb{P}(B^{\Delta_{S_f}} | C_\rho^{\Delta_{S_f}}) \mathbb{P}(C_\rho^{\Delta_{S_f}}) \nonumber \\
&\geq \frac{\mathbb{P}(T \leq \epsilon)}{\mathbb{P}(T \leq \eta\epsilon_0)} \cdot \gamma_2 \cdot (1 - 2e^{-2\rho^2n}) \label{eqn:10}
\end{align}
The first term in the final step of Equation~\ref{eqn:10} follows from the result in Corollary~\ref{cor:joint} where $T$ follows $\mathrm{BesselK}(k, \frac{\Delta_{S_f}}{\epsilon_0})$. It is the probability with which the Laplace mechanism $\mathcal{L}_{\epsilon_0}^{\Delta_{S_f}}$ satisfies $\epsilon$-differential privacy for a given value of sampled sensitivity.
% Last two terms in the final step of Equation~\ref{eqn:10} follow from the result in Lemma~\ref{lemma:main_1}.
\end{proof}

Probability of occurrence of event $A_{\epsilon_0}^{\Delta_{S_f}}$ calculated by accounting for both explicit and implicit sources of randomness gives the risk for privacy level $\epsilon$. Thus, the proof of Lemma~\ref{lemma:main_2} completes the proof for Theorem~\ref{thm:par_3}. 

Comparing the equations in Theorem~\ref{thm:par_3} and Lemma~\ref{lemma:main_2}, we observe that
\begin{equation}
\gamma_3 \triangleq \frac{\mathbb{P}(T \leq \epsilon)}{\mathbb{P}(T \leq \eta\epsilon_0)} \cdot \gamma_2
\label{eqn:gamma_3}
\end{equation}
The privacy at risk, as defined in Equation~\ref{eqn:gamma_3}, is free from the term that accounts for the accuracy of sampled estimate. If we know cumulative distribution of the sensitivity, we do not suffer from the uncertainty of introduced by sampling from the empirical distribution.

\end{document}